\documentclass[runningheads]{llncs}

\usepackage{graphics} % for pdf, bitmapped graphics files
\usepackage{epsfig} % for postscript graphics files
\usepackage{mathptmx} % assumes new font selection scheme installed
\usepackage{times} % assumes new font selection scheme installed
\usepackage{amsmath} % assumes amsmath package installed
\usepackage{amssymb}  % assumes amsmath package installed
\usepackage{cuted}
\usepackage{float}
\usepackage{color}
\usepackage{url}
\usepackage{bm}

\usepackage[ruled,vlined,linesnumbered,noresetcount]{algorithm2e}
\usepackage{tabularx}

\usepackage{caption,subcaption}
\DeclareMathOperator{\Tr}{Tr}
\DeclareMathOperator{\bTr}{\bar{T}r}

\begin{document}
\title{
Combating Informational Denial-of-Service (IDoS) Attacks: Modeling and Mitigation of Attentional Human Vulnerability 
}
\titlerunning{IDoS Attacks: Modeling and Mitigation of Attentional Human Vulnerability }

\author{Linan Huang \and
Quanyan Zhu
\thanks{This work is partially supported by grants SES-1541164, ECCS-1847056, CNS-2027884, and BCS-2122060 from National Science Foundation (NSF), and grant W911NF-19-1-0041 from Army Research Office (ARO).} 
}
\authorrunning{L. Huang and Q. Zhu}
\institute{Department of Electrical and Computer Engineering, New York University\\
 2 MetroTech Center, Brooklyn, NY, 11201, USA \\
\email{\{lh2328,qz494\}@nyu.edu}
}
\maketitle              % typeset the header of the contribution
\begin{abstract} %207
This work proposes a new class of proactive attacks called the Informational Denial-of-Service (IDoS) attacks that exploit the attentional human vulnerability. 
By generating a large volume of feints, IDoS attacks deplete the cognitive resources of human operators to prevent humans from identifying the real attacks hidden among feints. 
This work aims to formally define IDoS attacks, quantify their consequences, and develop human-assistive security technologies to mitigate the severity level and risks of IDoS attacks. 
To this end, we use the semi-Markov process to model the sequential arrivals of feints and real attacks with category labels attached in the associated alerts. 
The assistive technology strategically manages human attention by highlighting selective alerts periodically to prevent the distraction of other alerts. 
A data-driven approach is applied to evaluate human performance under different Attention Management (AM) strategies. 
Under a representative special case, we establish the computational equivalency between two dynamic programming representations to reduce the computation complexity and enable online learning with samples of reduced size and zero delays. 
%simplify the theoretical computation and the online learning. 
A case study corroborates the effectiveness of the learning framework. 
The numerical results illustrate how AM strategies can alleviate the severity level and the risk of IDoS attacks. 
Furthermore, the results show that the minimum risk is achieved with a proper level of intentional inattention to alerts, which we refer to as the \textit{law of rational risk-reduction inattention}.
%we Rational Risk-Reduction Inattention
%characterize the fundamental limits of the minimum severity level under all AM strategies and the maximum length of the inspection period to reduce the IDoS risks. 

\end{abstract}

\section{Introduction}
%Human is the weakest link in cybersecurity due to the acquired vulnerabilities (e.g., lack of security awareness and incentives) and innate vulnerabilities (e.g., bounded rationality and attention limitation). 
Human is the weakest link in cybersecurity due to their innate vulnerabilities, including bounded rationality and limited attention. 
These human vulnerabilities are difficult to mitigate through short-term training, rules, and incentives. As a result, sophisticated attacks, such as Advanced Persistent Threats (APTs) and supply-chain attacks, commonly exploit them to breach data and damage critical infrastructures. 
 Attentional vulnerabilities have been exploited by adversaries to create visual blindspots or misperceptions that can lead to erroneous outcomes. 
One way to exploit the attentional vulnerabilities is to stealthily evade the attention of human users or operators as we have seen in many cases of social engineering and phishing attacks. It is a passive approach where the attacker does not change the attention patterns of the human operators and intends to exploit the inattention to evade the detection. 
In contrast, a proactive attacker can strategically influence attention patterns. For example, an attacker can overload the attention of human operators with a large volume of feints and hide real attacks among them \cite{WinNT1}. 
This class of proactive attacks aims to increase the perceptual and cognitive load of human operators to delay defensive responses and reduce detection accuracy. 
%reduce their response time, detection accuracy, and performance efficiency.  
We refer to this class of attacks as the Informational Denial-of-Service (IDoS) attacks. 

IDoS is no stranger to us in this age of information explosion. We are commonly overloaded with terabytes of unprocessed data or manipulated information on online media. However, the targeted IDoS attacks on specific groups of people, e.g., security guards, operators at the nuclear power plant, and network administrators, can pose serious threats to lifeline infrastructures and systems. The attacker customizes attack strategies to targeted individuals or organizations to quickly and maximally deplete their human cognitive resources. 
As a result, common methods (e.g., set tiered alert priorities) to mitigate alert fatigue are insufficient under these targeted and intelligent attacks that generate massive feints strategically. 
%IDoS attacks have significant consequences, including distorted perceptions and erroneous security decisions. They also threaten the security of human-machine interaction applications, including semi-autonomous driving and industrial control systems. 
There is a need to understand this phenomenon, quantify its consequence and risks, and develop new mitigation methods. 
%Although feints have been extensively used in physical worlds, including sports and wars, feint-induced IDoS attacks in cyberspace have not been formally defined and quantified. 
In this work, we establish a probabilistic model to formalize the definition of IDoS attacks, evaluate their severity levels, and assess the induced cyber risks. 
The model captures the interaction among attackers, human operators, and assistive technologies as highlighted by the orange, green, and blue backgrounds, respectively, in Fig. \ref{fig:AHM}. 
\begin{figure}[h]
\centering
  \includegraphics[width=1\linewidth]{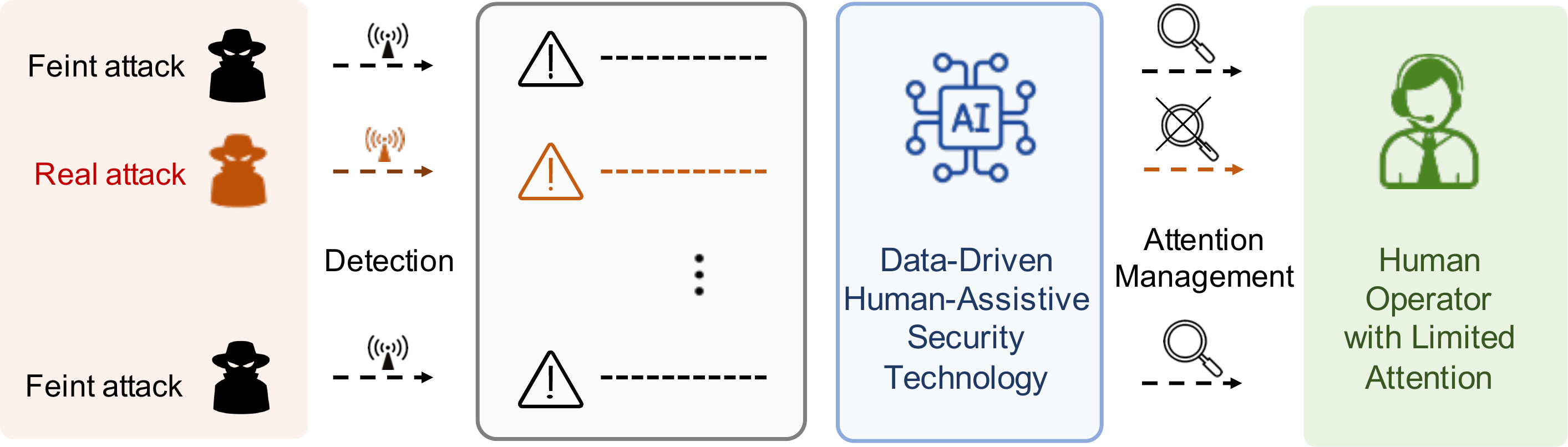}
  \caption{
Interaction among IDoS attacks, 
  human operators, and assistive technologies. 
  }
\label{fig:AHM}
\end{figure}

Attackers generate feints and real attacks that trigger alerts of detection systems. 
Due to the detection imperfectness, human operators need to inspect these alerts in detail to determine the attacks' types, i.e., feint or real, and take responsive security decisions.  
The accuracy of the security decisions depends on the inspection time and the operator's sustained attention without distractions. 
The large volume of feints exerts an additional cognitive load on each human operator and makes it hard to focus on each alert, which can significantly decrease the accuracy of his security decisions and increase cyber risks. 
Accepting the innate human vulnerability, we aim to develop assistive technologies to compensate for the human attention limitation. 
Evidence from the cognitive load theory \cite{wickens2015engineering} has shown that divided attention to multiple stimuli can degrade the performance and cost more time than responding to these stimuli in sequence. 
Hence, we design the \textit{Attention Management} (AM) strategies to intentionally make some alerts inconspicuous so that the human operator can focus on the other alerts and finish the inspection with less time and higher accuracy. 
We further define risk measures to evaluate the inspection results, which serves as the stepping stone to designing adaptive AM strategies to mitigate attacks induced by human vulnerabilities. 

Due to the unpredictability and complexity of human behaviors, cognition, and reasoning, it is challenging to create an exact human model of the IDoS attack response. 
Therefore, we provide a probabilistic characterization of human decisions concerning AM strategies and other observable features from the alerts. By assuming a sequential arrival of attacks with semi-Markov state transitions, we conduct a data-driven approach to evaluate the inspection results in real-time. 
Under a mild assumption, we prove the \textit{computational equivalency} between two Dynamic Programming (DP) representations to simplify the value iteration and the Temporal-Difference (TD) learning process. 
Numerical results corroborate the effectiveness of learning by showing the convergence of the estimated value to the theoretical value. 
Without an AM strategy, we show that both the severity level and the risk of IDoS attacks increase with the product of the arrival rate and the detection threshold. 
With the assistance of AM strategies, 
%we characterize two fundamental limits, i.e., the minimum severity level and the maximum length of the inspection period. 
%Then, 
we illustrate how different AM strategies can alleviate the severity level of IDoS attacks. Concerning the IDoS risks, we illustrate the tradeoff between 
the quantity and quality of the inspection, which leads to a meta-principle referred to as the \textit{law of rational risk-reduction inattention}. 
%a high quantity of inspection with low accuracy and a low quantity of inspection with high accuracy. 
%focusing on a small number of alerts for high accuracy and a large number of alerts to avoid misdetection. 

\subsection{Related Works}
\subsubsection{Human Vulnerability in Cyber Space} 
%Human vulnerabilities in cyberspace can be classified into innate vulnerabilities and acquired vulnerabilities. 
Attacks that exploit human vulnerabilities, e.g., insider threats and social engineering, have raised increasing concerns in cybersecurity. 
Previous works have focused to design security rules \cite{casey2016compliance} and incentives  \cite{huang2021duplicity} to increase human employees' compliance and elicit desirable behaviors. 
However, compared to the lack of security awareness and incentives, some human vulnerabilities (e.g., attention limitation and bounded rationality) cannot be altered or controlled. Thus, we need to design assistive technologies to compensate for the `unpatchable' human vulnerabilities. 
In \cite{huang2021inadvert}, adaptive attention enhancement strategies have been developed to engage users' attention and maximize the rate of phishing recognition. 
Compared to \cite{huang2021inadvert} that defends against stealthy attacks and the exploitation of inattention, this work combats proactive attackers that overload human attention. %by pivoting human attention on a small number of selected alerts. 
% This work aims to create an assistive technology oppositely (WHAT DOES IT MEAN "OPPOSITELY"?). 
% Rather than enhancing the attention of human operators, we intentionally pivot their attention on selected alerts so that they have more time to focus on other alerts. (NOT CLEAR WHAT IT MEANS!)

%\subsubsection{Information Overload and Alert Fatigue} 
%The explosion of information is a double-edged sword. If 
%Information overload, or infoxication has raised 

\subsubsection{Data-Driven Approach for Security and Resilience}
As more data becomes available, data-driven approaches have been widely used to create cyber situational awareness and enhance network security and resilience \cite{huang2021reinforcement}, e.g., Bayesian learning for parameter uncertainty \cite{huang2020dynamic,huang9494340} and Q-learning for honeypot engagement \cite{huang2019adaptive}. 
The authors in \cite{zhao2019bidirectional} have studied the detection of feint attacks by a few-shot
deep learning algorithm. However, they have modeled feints as multi-stage attacks and focused on detecting the revised causal relationship. Here, we focus on how feints affect human operators' cognitive resources and the consequent security decisions. 
The TD learning method helps address the long-standing challenge of human modeling and further enables us to evaluate human performance efficiently and robustly. 

%\subsubsection{Data-Driven Approach to Enhance Security and Resilience} 

%https://www.tobiipro.com/learn-and-support/learn/eye-tracking-essentials/types-of-eye-movements/
%https://en.wikipedia.org/wiki/Visual_field
%https://en.wikipedia.org/wiki/Peripheral_vision
%https://www.tobiipro.com/learn-and-support/learn/eye-tracking-essentials/

\subsection{Notations and Organization of the Paper}
We summarize notations in Table \ref{table:notation}.  
The rest of the paper is organized as follows. 
Section \ref{sec:IDoS model} introduces the system modeling for IDoS attacks, alert generations, and the inspections of human operators. 
Based on the system model, we present a Semi-Markov Process (SMP) model in Section \ref{sec:SMP} to evaluate human performance, the severity level, and the risks of IDoS attacks. 
We present a case study in Section \ref{sec:casestudy} to corroborate our results and  Section \ref{sec:conclusion} concludes the paper. 

\begin{table}[th]
\centering
\caption{Summary of variables and their meanings. 
\label{table:notation}}
\begin{tabularx}{\columnwidth}{X l} %\textwidth %{|l|X|} %The column type X takes all the space left from other columns till \textwidth or whatever you specify in the first argument.
     \hline
\textbf{Variable} &  \textbf{Meaning} \\ \hline
$t^k\in [0,\infty)$  &  Arrival time of the $k$-th attack \\
$\tau^k=t^{k+1} - t^k\in [0,\infty)$       &  Time duration between $k$-th and $(k+1)$-th attack\\
$\tau_{IN}^{h,m}:=\sum_{{k}'=hm}^{hm+m-1} \tau^{{k}'}$ & Inspection time at inspection stage $h\in  \mathbb{Z}^{0+}$  \\
$w^k\in \mathcal{W}:=\{w_{FE},w_{RE},w_{UN}\}$       &  Security decision at attack stages $k\in \mathbb{Z}^{0+}$\\
$a_m\in \mathcal{A}$ & Attention management strategy of period $m\in \mathbb{Z}^+$ \\
$\theta^k\in \Theta:=\{\theta_{FE},\theta_{RE}\}$       &  Attack's type at  attack stages $k\in \mathbb{Z}^{0+}$\\
 $\bar{\theta}^{h}:=[\theta^{hm},\cdots,\theta^{hm+m-1}]$  & Consolidated type at inspection stage $h\in  \mathbb{Z}^{0+}$  \\ 
$s^k\in \mathcal{S}$       &  Alert's category label at attack stages $k\in \mathbb{Z}^{0+}$\\
$x^h:=[s^{hm},\cdots,s^{hm+m-1}]$ & Consolidated state at inspection stage $h\in \mathbb{Z}^{0+}$ \\
$\Tr(s^{k+1}|s^k;\theta^k)$ & Transition probability from $s^k$ to $s^{k+1}$ under attack type $\theta^k$ \\
$\bTr(x^{h+1}|x^h ;\bar{\theta}^{h} )$  & Transition function of the consolidated state \\
% $1-\hat{p}_{CD}(x^h,a_m)$ &  Consolidated severity level of IDoS attacks under $x^h$ and $a_m$.  \\
%$\Pr(w^{k}|s^k, a_m;\theta^k)$ & Decision probability  \\
%$\mathcal{N}(s^k,\theta^k)$ & Set of $N$ thresholds that affect the decision probability \\
\hline
\end{tabularx}
\end{table}

\section{System Modeling of Informational Denial-of-Service Attacks}
\label{sec:IDoS model}
In Section \ref{sec:motivatingexample}, we present a high-level structure of the Informational Denial-of-Service (IDoS) attacks and use a motivating example to illustrate their causes, consequences, and mitigation methods.  
Then, we introduce the system modeling of sequential arrivals of alerts that are triggered by feints and real attacks in Section \ref{sec:attack}. 
The manual inspection and the attention management strategies are introduced in Section \ref{sec:inspectionAM}. 
Human operators inspect each alert in real-time to determine the associated attack's hidden type. 
Meanwhile, the assistive technology automatically designs and implements the optimal attention management strategy to compensate for human attention limitations. 
%The real-time inspection provided by human inspectors and the attention management strategy provided by the assistive technology are introduced in Section \ref{sec:inspectionAM} to determine the attack's hidden type and compensate for human attention limitations, respectively. 

\subsection{High-Level Abstraction and Motivating Example}
\label{sec:motivatingexample}
As shown in Fig. \ref{fig:abstraction}, there is an analogy between the Denial-of-Service (DoS) attacks in communication networks and the Informational Denial-of-Service (IDoS) attacks in the human-in-the-loop systems. Both of them achieve their attack goals by exhausting the limited resources. 
DoS attacks happen when the attacker generates a large number of superfluous requests to deplete the computing resource of the targeted machine and prevent the fulfillment of legitimate services. 
Analogously, IDoS attacks create a large amount of unprocessed information to deplete cognitive resources of human operators and prevent them from acquiring the knowledge contained in the information. 
\begin{figure}[h]
\centering
  \includegraphics[width=.74\linewidth]{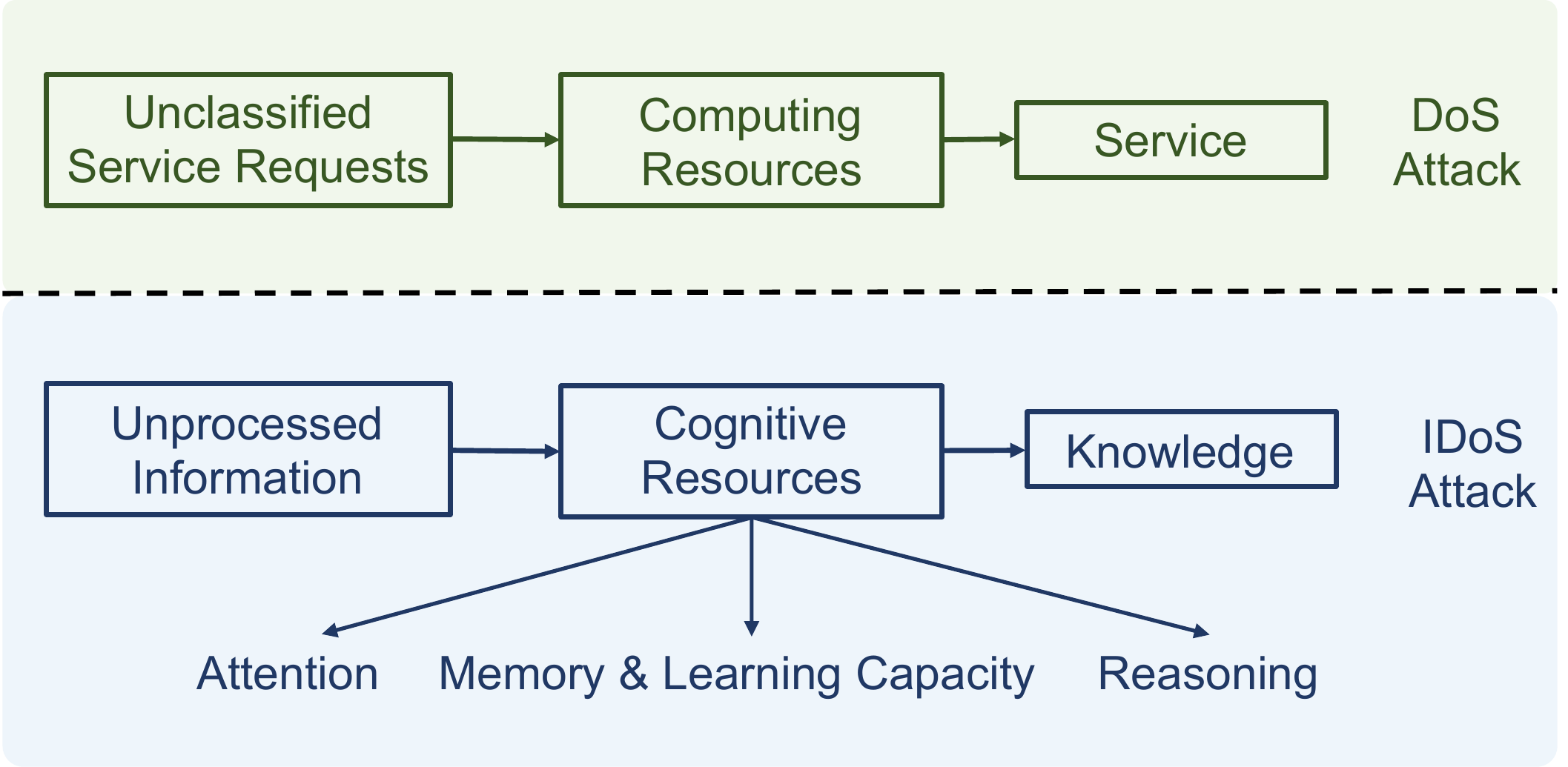}
  \caption{
The service request fulfillment process under DoS attacks and the information processing flows under IDoS attacks in green and blue backgrounds, respectively. 
  }
\label{fig:abstraction}
\end{figure}
We list several assailable cognitive resources under IDoS attacks as follows. 
\begin{itemize}
    \item \textbf{Attention}: 
    Paying sustained attention to acquire proper information is costly. 
    From an economic perspective, inattention occurs when the cost of information acquisition is lower than the attention cost measured by the information entropy \cite{sims2003implications}. 
    %Attention is the behavioral and cognitive process of selectively concentrating on a discrete aspect of information, whether considered subjective or objective, while ignoring other perceivable information. 
 IDoS attacks generate feints to distract the human from the right information. An excessive number of feints prohibit the human from process any information. 
     \item \textbf{Memory and Learning Capacity}:  Humans have limited memory and learning capacity. Humans cannot remember the details or learn new things if there is an information overload \cite{wickens2015engineering}. 
    \item \textbf{Reasoning}: Human decision-making consumes a large amount of energy, which is one of the reasons why we have two modes of thought \cite{kahneman2011thinking} (`system $1$' thinking is fast, instinctive, and emotional; while `system $2$' thinking is slower and more logical). IDoS attacks can exert a heavy cognitive load to prevent humans from deliberative decisions that use the `system $2$' thinking. Moreover, evidence shows the \textit{paradox of choice} \cite{schwartz2004paradox}; i.e., rich choices can bring anxiety and prevent humans from making any decisions. 
\end{itemize}

When these cognitive resources are exhausted, the information cannot be processed correctly and timely and serves as noise that leads to \textit{alert fatigue} \cite{ban2021combat}. 
We use operators in the control room of nuclear power plants as a stylized  example to illustrate the consequences of IDoS attacks and motivate the need for the security technology to assist human operators against IDoS attacks. 
%illustrate the above phenomena. motivating
\begin{figure}[h]
\centering
  \includegraphics[width=.5\linewidth]{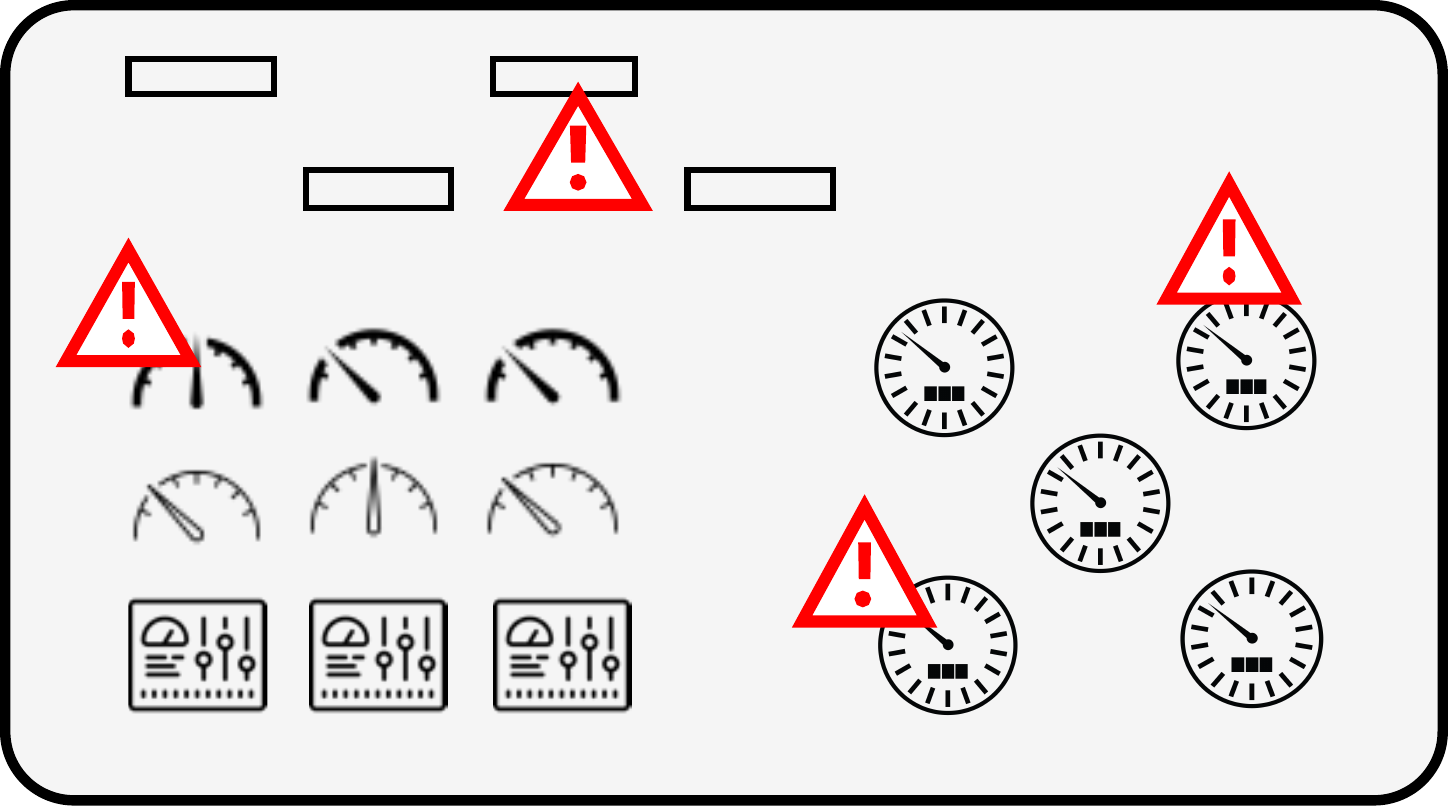}
  \caption{
  A stylized example of the monitor screen for operators in the control room of nuclear power plants. The red triangles represent warnings and security messages. 
  }
\label{fig:example}
\end{figure}
In Fig. \ref{fig:example}, a monitor screen contains meters that show the real-time readings of the temperature, pressure, and flow rate in a nuclear power plant. 
Based on the pre-defined generation rules, warnings and messages pop up at different locations. 
Due to the complexity of the nuclear control system, the inspection of these alerts consumes the operator's time and cognitive resources. 
The attempt to inspect all alerts and the constant switching among them can lead to missed detection and erroneous behaviors. 
If the alerts are generated strategically by attacks, they may further mislead humans to take actions in the attacker's favor; e.g., focusing on feints and ignoring the real attacks that hide among feints. 

%(Expand the following into a paragraph)
One way to mitigate IDoS attacks is to train the operators or human users to deal with the information overload and remain vigilant and productive under a heavy cognitive load. 
However, attentional training can be time-consuming and the effectiveness is not guaranteed. 
The second method is to recruit more human operators to share the information load. It would require the coordination of the operator team and can incur additional costs of human resources. 
The third method is to develop assistive technologies to rank and filter the information to alleviate the cognitive load of human operators. It would leverage past experiences and data analytics to pinpoint and prioritize critical alerts for human operators to process. 
The first two methods aim to increase the capacity or the volume of the cognitive resources in Fig. \ref{fig:abstraction}. 
The third method pre-processes the information so that it adapts to the capacity and characteristics of cognitive resources. 

% Mitigation methods of IDoS include the following: 
% \begin{itemize}    
%     \item Train the operators to deal with the information overload. 
%     \item Coordinate a team of human operators to share the information load. 
%     \item Develop assistive technology to rank and filter the information to alleviate the cognitive load of human operators. (visual aids can be used to )
% \end{itemize}
%This work focuses on the security assistive technology. 

\subsection{Sequential Arrivals of Alerts Triggered by Feints and Real Attacks}
\label{sec:attack}
In this work, we focus on the temporal aspect of the alerts (i.e., the frequency and duration of their arrivals). 
The future work will incorporate their spatial locations on the monitor screen as shown in Fig. \ref{fig:example}. 
As highlighted by the orange background in Fig. \ref{fig:newMDP}, attacks arrive sequentially at time $t^k, k\in \mathbb{Z}^{0+}$ where $t^0=0$. 
Let $\tau^k:=t^{k+1}-t^k\in [0,\infty)$ be the inter-arrival time between the $(k+1)$-th attack and the $k$-th attack for all $k\in \mathbb{Z}^{0+}$. 
We refer to the $k$-th attack equivalently as the one at \textit{attack stage} $k\in \mathbb{Z}^{0+}$. 
 \begin{figure}[h]
\centering
\includegraphics[width=0.9 \textwidth]{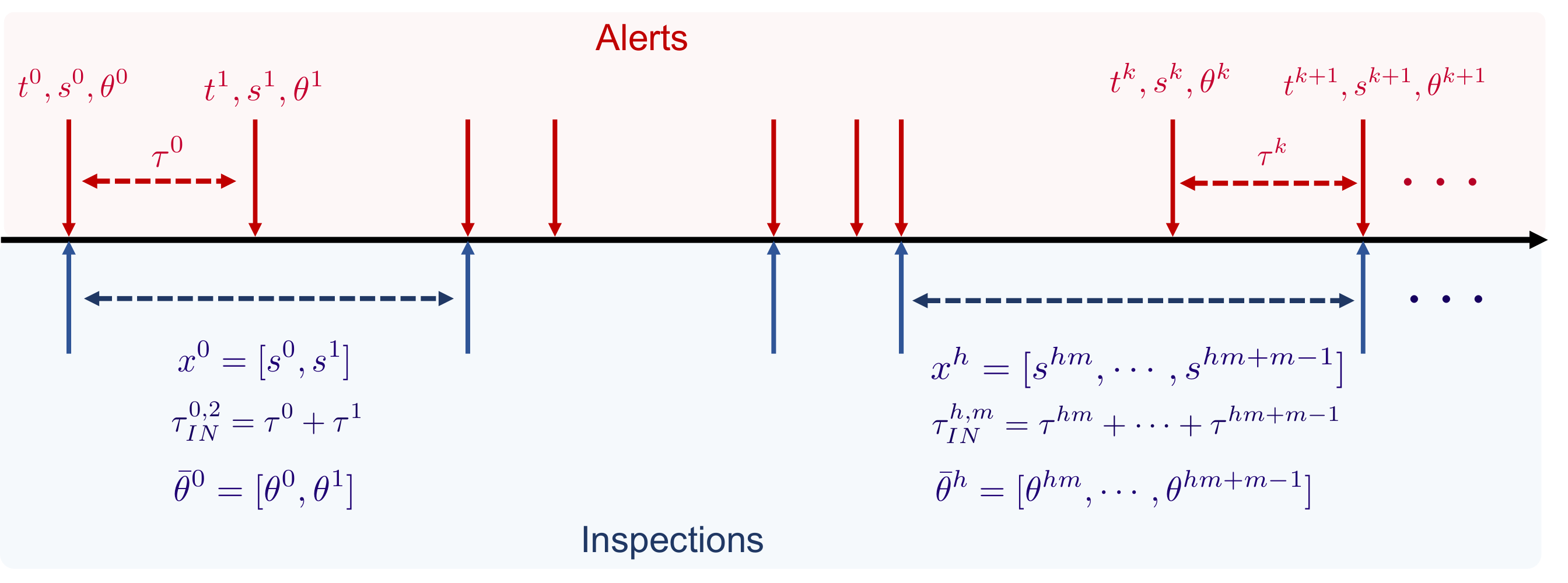}
\caption{ 
The sequential arrival of alerts at \textit{attack stage} $k\in \mathbb{Z}^{0+}$
and 
the periodic manual inspections at \textit{inspection stage} $h\in  \mathbb{Z}^{0+}$ under AM strategy $a_m\in \mathcal{A}$ where $m=2$.  
}
\label{fig:newMDP}
\end{figure}  

Each attack can be either a feint (denoted by $\theta_{FE}$) or a real attack (denoted by $\theta_{RE}$) with probability $b_{FE}\in [0,1]$ and $b_{RE}\in [0,1]$, respectively, where $b_{FE}+b_{RE}=1$. 
We assume that both types of attacks trigger alerts with the same time delay. % trigger alerts with different probability. %the probability does not matter, we cannot do anything if we do not detect them. We can only respond to alerts. 
%mis-detection does not matter and false alert from normal user make it worse. 
%The normal users may also accidentally trigger alerts. %this will make the tau distribution hard to compute. 
Thus, there is a one-to-one mapping between the sequence of attacks and alerts, and we can consider the zero delay time without loss of generality. 
The alerts cannot reflect the \textit{attack's type} denoted by $\theta^k\in \Theta:=\{\theta_{FE},\theta_{RE}\}$ at all attack stages $k\in \mathbb{Z}^{0+}$. 
However, the alerts can provide human operators with a \textit{category label} from a finite set $\mathcal{S}$ based on observable features or traces of the associated attacks, e.g., the attack locations as shown in Section \ref{sec:casestudy}. 
%https://www.cisecurity.org/cybersecurity-threats/alert-level/
% Moreover, the set of category labels can indicate the level of lethality, criticality, or defensibility as follows. 
% \begin{itemize}
%     \item \textbf{Lethality}: how much damage can the attack incur? For example, attacks with root privilege can inflict severer damage than the ones with limited access. 
%     \item \textbf{Criticality}: how important is the target of the attack? For example, cloud servers are more critical than client servers. 
%     \item \textbf{Defensibility}: how likely can the system counteract the attack? For example, operating systems with applicable patches can better prevent attacks than the ones without anti-virus software. 
%     %\item \textbf{Resiliency}: how much and how quickly can the system recover from the attack?  For example, systems with distributed data storage restore the data damage more quickly than  the ones with centralized storage. 
% \end{itemize}
% %% we obtain the risk=Lethality+Criticality-Defensibility-Resiliency %this could be the journal extension where we consider adaptive strategy. When high risk alert come, we should switch if the current one is of low risk. 
% %For example, the category label can indicate the target of the attack or the severity of the attack's potential outcomes. 
We denote the alert's category label at attack stage $k\in \mathbb{Z}^{0+}$ as $s^k\in \mathcal{S}$. 

%Based on the severity and the time-sensitivity of their potential outcomes, these alerts categorize attacks into $N_s$ categories represented by a finite set $\mathcal{S}$. 
%We assume the classification is correct. 
%For example, $\mathcal{S}=\{s_1,s_2,s_3\}$ can contain $N_s=3$ categories of alerts, where $s_1,s_2,s_3$ indicate attacks of low, medium, and high risks, respectively. 
%The alerts provide the human operator with the $k$-th attack's \textit{severity level} denoted by $s^k\in \mathcal{S}$ for all $k\in \mathbb{Z}^{0+}$. 

\subsection{Manual Inspection and Attention Management}
\label{sec:inspectionAM}
%we may extend it to multiple operators, then it becomes a sequential matching problem. 
Since an alert does not directly reflect whether the attack is feint or real, human operators need to inspect the alert to determine the hidden type, which leads to three \textit{security decisions}: the attack is feint (denoted by $w_{FE}$), the attack is real (denoted by $w_{RE}$), or the attack's type is unknown (denoted by $w_{UN}$). 
We use $w^k\in \mathcal{W}:=\{w_{FE},w_{RE},w_{UN}\}$ to denote the human operator's security decision of the $k$-th alert. 
Each human operator has limited attention and cannot inspect multiple alerts simultaneously. Moreover, the human operator requires sustained attention on an alert to make an accurate security decision. 
Frequent alert pop-ups can distract humans from the current alert inspection and result in \textit{alert fatigue} and the \textit{paradox of choice} as illustrated in Section \ref{sec:motivatingexample}. 
To compensate for the human's attention limitation, we can intentionally make some alerts less noticeable, e.g.,  without sounds or in a light color. 
Then, the human can pay sustained attention to the alert currently under inspection. 
These inconspicuous alerts can be assigned to other available inspectors with an  additional cost of human resources. 
If these alerts are time-insensitive, they can also be queued and inspected later by the same operator at his convenience. 
However, in practice, the number of alerts usually far exceeds the number of available inspectors, and the alerts cannot tolerate delay. 
Then, these alerts are dismissed as a tradeoff for the timely and accurate inspection of the other highlighted alerts. 
In this case, these inconspicuous alerts are not inspected and automatically assigned the security decision $w_{UN}$. 

In this paper, we focus on the class of \textit{Attention Management (AM) strategies}, denoted by $\mathcal{A}:=\{a_m\}_{m\in \mathbb{Z}^+}$, that highlight alerts periodically to engage operators in the alert inspection. 
We assume that the human operator can only notice and inspect an alert when it is highlighted. 
Then, AM strategy $a_m\in \mathcal{A}$ means that the human operator inspects the alerts at attack stages $k=hm, h\in  \mathbb{Z}^{0+}$. 
%, where $H:=\left \lfloor{K/M}\right \rfloor$ is the greatest integer less than $K/M$. 
%inspects every $m\in\mathcal{M}:=\{1,2,\cdots,M\}$ alerts and ignore the rest. 
We refer to the attack stages during the $h$-th inspection as the \textit{inspection stage} $h\in  \mathbb{Z}^{0+}$. 
Then, under AM strategy $a_m\in \mathcal{A}$, each inspection stage contains $m$ attack stages as shown in the blue background of Fig. \ref{fig:newMDP}. 
The $h$-th inspection has a duration of $\tau_{IN}^{h,m}:=\sum_{{k}'=hm}^{hm+m-1} \tau^{{k}'}$ for all $h\in  \mathbb{Z}^{0+}$.

\subsubsection{Decision Probability with $N$ Thresholds}
The human operator's security decision depends on the attack's type, the category label, and the AM strategy. 
We refer to $\Pr(w^k|s^k, a_m;\theta^k)$ as the \textit{decision probability}; i.e., the probability of human making decision $w^k\in \mathcal{W}$ when the attack's type is $\theta^k\in \Theta$, the category label is $s^k\in \mathcal{S}$,  and the AM strategy is $a_m\in \mathcal{A}$. 
As a probability measure, the decision probability satisfies 
$
   \sum_{w^k \in \mathcal{W}} \Pr(w^{k}|s^k, a_m;\theta^k)=1, \forall \theta^{k}\in \Theta, \forall s^k\in \mathcal{S}, \forall a_m\in \mathcal{A}. 
$
%For simplicity, let assume the transition probability is not affected by $\theta$. otherwise, the cumulative utility is a function of $\theta$, and it is harder to compute. 

At attack stages where alerts are inconspicuous, i.e., for all $k\neq hm, h\in  \mathbb{Z}^{0+}$, the security decision $w^k$ is $w_{UN}$ with probability $1$; i.e., for any given inspection policy $a_m\in \mathcal{A}$, we have 
$
    \Pr(w^k|s^k, a_m;\theta^k)=\mathbf{1}_{ \{ w^k=w_{UN} \} },\forall s^k\in \mathcal{S}, \forall w^k\in \mathcal{W}, \forall \theta^{k}\in \Theta,  \forall k\neq hm, h\in  \mathbb{Z}^{0+}. 
$
% \[
% \Pr(w^k|\theta^k,a_m)=
% \begin{cases}
% 1 &\text{if } w^k=w_{UN}\\
% 0 &\text{if } w^k\in \{w_{FE},w_{RE}\}
% \end{cases}
% ,\forall \theta^{k}\in \Theta, \forall k\neq hm, h\in \{0,1,2,\cdots\}. 
% \]
At attack stages of highlighted alerts, i.e., for all $k=hm, h\in  \mathbb{Z}^{0+}$, the human operator inspects the $h$-th alert for a duration of $\tau_{IN}^{h,m}$. 
At each inspection stage $h$, a longer period length $m$ induces a longer inspection time $\tau_{IN}^{h,m}=\sum_{{k}'=hm}^{hm+m-1} \tau^{{k}'}$. 
Based on the IDoS model in Section \ref{sec:IDoS model}, different AM strategies only affect the inspection time. 
Thus, we can rewrite the decision probability $\Pr(w^{k}|s^k, a_m;\theta^k)$ as $\Pr(w^{k}|s^k, \tau_{IN}^{h,m};\theta^k)$ at attack stages $k=hm, h\in  \mathbb{Z}^{0+}$. 

Adequate inspection time $\tau_{IN}^{h,m}$ leads to an accurate security decision. 
In this work, we assume that the probability of correct decision-making can be approximated by an increasing step function of the inspection time
%increase of the inspection time has a threshold impact on the decision probability 
as shown in Fig. \ref{fig:thresholdprob}. 
 \begin{figure}[h]
\centering
\includegraphics[width=0.9 \textwidth]{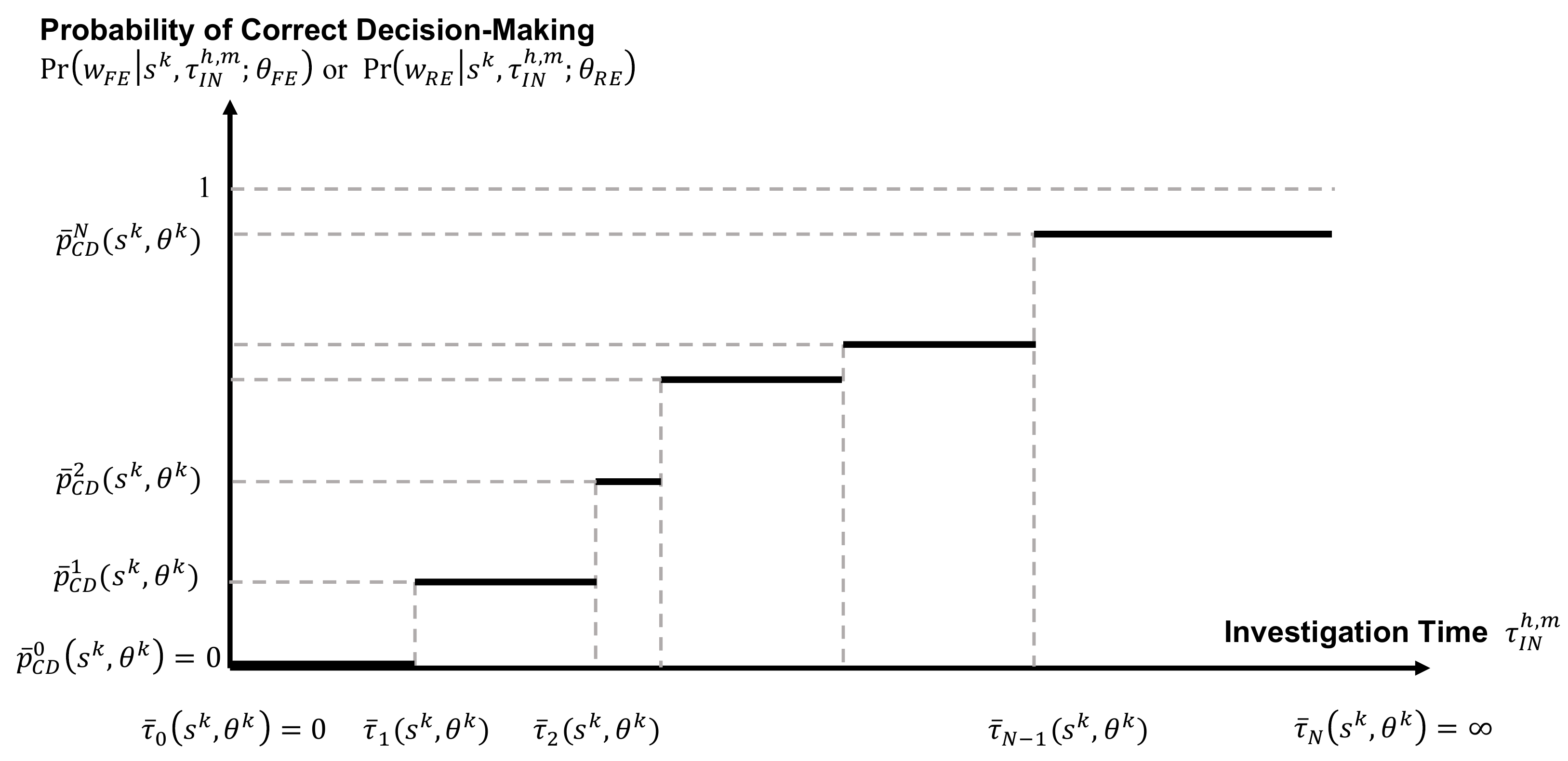}
\caption{ 
The probability of the human operator making correct security decisions, i.e., $\Pr(w_{FE}|s^k,\tau_{IN}^{h,m};\theta_{FE})$ and $\Pr(w_{RE}|s^k,\tau_{IN}^{h,m};\theta_{RE})$, is approximated as an increasing step function of the inspection time $\tau_{IN}^{h,m}$ at inspection stage $h\in  \mathbb{Z}^{0+}$. 
%has a threshold impact on $\Pr(w_{FE}|s^k,\tau_{IN}^{h,m};\theta_{FE})$ and $\Pr(w_{RE}|s^k,\tau_{IN}^{h,m};\theta_{RE})$, i.e., the probability of the human operator making correct security decisions. 
}
\label{fig:thresholdprob}
\end{figure}  
That is, $N+1$ thresholds divide  the support of the random variable $\tau_{IN}^{h,m}$, i.e., $[0,\infty)$, into $N$ regions where the probability of correct security decisions increases. 
We can increase the number of thresholds, i.e., the value of $N$, to improve the accuracy of the approximation. 
For each $s^k\in \mathcal{S}$ and $\theta^k\in \Theta$, we denote the corresponding $N$ thresholds as $\bar{\tau}_n(s^k,\theta^k)\in \mathcal{N}(s^k,\theta^k), n\in \{0,1,\cdots,N\}$, where $\mathcal{N}(s^k,\theta^k)$ is a finite set, $\bar{\tau}_0(s^k,\theta^k)=0$, $\bar{\tau}_N(s^k,\theta^k)=\infty$, and $\bar{\tau}_0(s^k,\theta^k)<\bar{\tau}_N(s^k,\theta^k)<\bar{\tau}_2(s^k,\theta^k)<\cdots<\bar{\tau}_N(s^k,\theta^k)$. 
If $\tau_{IN}^{h,m}$ belongs to the region $n\in \{0,1,\cdots,N\}$, i.e.,  $\bar{\tau}_{n-1}(s^k,\theta^k)<\tau_{IN}^{h,m}<\bar{\tau}_{n}(s^k,\theta^k)$, 
then the decision probabilities under %$\theta^k\in \Theta$ 
$\theta_{FE}$ and $\theta_{RE}$ are represented as \eqref{eq:detection probability FE} and \eqref{eq:detection probability RE}, respectively, 
%for all $k=hm, h\in  \mathbb{Z}^{0+}$, has the following form: 
% \begin{equation}
% \label{eq:detection probability}
% \Pr(w^k|s^k,\tau_{IN}^{h,m};\theta^k)=
% \begin{cases}
% \bar{p}_{CD}^{n-1}(s^k,\theta^k)\in [0,1] &\text{if } w^k=w_{FE} \\
% \bar{p}_{ID}^{n-1}(s^k,\theta^k)\in [0,1] &\text{if } w^k=w_{RE}   \\
% 1-\bar{p}_{CD}^{n-1}(s^k,\theta^k)-\bar{p}_{ID}^{n-1}(s^k,\theta^k) &\text{if } w^k=w_{UN} \\
% \end{cases}
% %, \tau_{IN}^k > \bar{\tau}
% \end{equation}
\begin{equation}
\label{eq:detection probability FE}
\Pr(w^k|s^k,\tau_{IN}^{h,m};\theta_{FE})=
\begin{cases}
\bar{p}_{CD}^{n-1}(s^k,\theta_{FE})\in [0,1] &\text{if } w^k=w_{FE} \\
\bar{p}_{ID}^{n-1}(s^k,\theta_{FE})\in [0,1] &\text{if } w^k=w_{RE}   \\
1-\bar{p}_{CD}^{n-1}(s^k,\theta_{FE})-\bar{p}_{ID}^{n-1}(s^k,\theta_{FE}) &\text{if } w^k=w_{UN} \\
\end{cases}
\end{equation}
and 
\begin{equation}
\label{eq:detection probability RE}
\Pr(w^k|s^k,\tau_{IN}^{h,m};\theta_{RE})=
\begin{cases}
\bar{p}_{CD}^{n-1}(s^k,\theta_{RE})\in [0,1] &\text{if } w^k=w_{RE}  \\
\bar{p}_{ID}^{n-1}(s^k,\theta_{RE})\in [0,1] &\text{if } w^k=w_{FE} \\
1-\bar{p}_{CD}^{n-1}(s^k,\theta_{RE})-\bar{p}_{ID}^{n-1}(s^k,\theta_{RE}) &\text{if } w^k=w_{UN} \\
\end{cases}
\end{equation}
In both \eqref{eq:detection probability FE} and \eqref{eq:detection probability RE}, the first and second cases represent the probability of making correct and incorrect security decisions, respectively. 
The third case represents the probability that the human operator is uncertain about the attack's type and needs more time to inspect.  
%In this work, we assume that the human operator is prudent and does not make incorrect security decisions, i.e., $\Pr(w_{RE}|s^k,\tau_{IN}^{h,m};\theta_{FE})=0$ and $\Pr(w_{FE}|s^k,\tau_{IN}^{h,m};\theta_{RE})=0$. 
A longer inspection time has two impacts: 
\begin{itemize}
    \item Increases the probability of making correct security decisions, i.e., $0=\bar{p}_{CD}^{0}(s^k,\theta^k)\leq \bar{p}_{CD}^{1}(s^k,\theta^k)\leq \cdots \leq \bar{p}_{CD}^{N}(s^k,\theta^k)\leq 1$, for any given $s^k\in \mathcal{S}$ and $\theta^k\in \Theta$. 
    \item Decreases the probability of incorrect security decisions, i.e., $0\leq \bar{p}_{ID}^{N}(s^k,\theta^k)\leq \bar{p}_{ID}^{N-1}(s^k,\theta^k)\leq \cdots \leq \bar{p}_{ID}^{0}(s^k,\theta^k)\leq 1$, for any given $s^k\in \mathcal{S}$ and $\theta^k\in \Theta$. 
\end{itemize}

% a generalised  sigmoid function of the inspection time $\tau_{IN}^k$, i.e., 
% %$f(x)=2/(1+e^{-(x-\bar{\tau})})-1$ where \bar{\tau} from 0 to infinity. see fig
% \begin{equation}
% \label{eq:detection probability FE}
% \Pr(w^k|\theta_{FE},s^k,a_m)=
% \begin{cases}
% 2/(1+e^{-(\tau_{IN}^k-\bar{\tau}(s^k))})-1 &\text{if } w^k=w_{FE} \\
% 0 &\text{if } w^k=w_{RE}  \\
% 2-2/(1+e^{-(\tau_{IN}^k-\bar{\tau}(s^k))}) &\text{if } w^k=w_{UN} \\
% \end{cases}
% %, \tau_{IN}^k > \bar{\tau}
% \end{equation}
% and 
% \begin{equation}
% \label{eq:detection probability RE}
% \Pr(w^k|\theta_{RE},s^k,a_m)=
% \begin{cases}
% 0 &\text{if } w^k=w_{FE} \\
% 2/(1+e^{-(\tau_{IN}^k-\bar{\tau}(s^k))})-1 &\text{if } w^k=w_{RE}  \\
% 2-2/(1+e^{-(\tau_{IN}^k-\bar{\tau}(s^k))}) &\text{if } w^k=w_{UN} \\
% \end{cases}
% %, \tau_{IN}^k > \bar{\tau}
% \end{equation}
% Note that the above two are conditional probability on $\tau_{IN}$. 

%We assume that the operator is not allowed to ignore more than $M$ successive alerts. 
%Thus, the human operator can choose his inspection time from the following finite set $\mathcal{A}:=\{\}$
%We denote the inspection time for the $k$-th alert as $a^k\in [0,\infty)$. 

\section{Semi-Markov Process Model for Performance Evaluation}
\label{sec:SMP}
We assume that the category label of the sequential attacks follows a semi-Markov process based on the attack's type where $\Tr(s^{k+1}|s^k;\theta^k)$ represents the \textit{transition probability} from $s^k\in\mathcal{S}$ to $s^{k+1}\in\mathcal{S}$ when the attack's type is $\theta^k\in \Theta$ at attack stage $k\in \mathbb{Z}^{0+}$. 
As a probability measure, the transition probability satisfies
$\sum_{s^{k+1}\in \mathcal{S}} \Tr(s^{k+1}|s^k;\theta^k)=1, \forall s^{k}\in \mathcal{S}, \forall \theta^k\in \Theta$. 
The inter-arrival time $\tau^k$ is a continuous random variable with a Probability Density Function (PDF) denoted by $z(\cdot|s^k;\theta^k)$.

\subsection{Consolidated State and Consolidated Cost}
Since the inspection is made every $m$ attack stages, we define the \textit{consolidated state} $x^h:=[s^{hm},\cdots,s^{hm+m-1}]\in \mathcal{X}:= \mathcal{S}^m$ that consists of the category labels of $m$ successive alerts at inspection stage $h\in  \mathbb{Z}^{0+}$. 
Analogously, we define the \textit{consolidated type} $\bar{\theta}^{h}:=[\theta^{hm},\cdots,\theta^{hm+m-1}]\in \bar{\Theta}:= \Theta^m $. % as shown in Fig. \ref{fig:newMDP}. 
Then, we denote the transition function of the consolidated state as $\bTr(x^{h+1}|x^h ;\bar{\theta}^{h} )$, which is also Markov as shown below. 
\begin{equation}
\label{eq:bTr}
\resizebox{1 \hsize}{!}{$
\begin{split}
       & \Pr(x^{h+1}|x^h,\cdots,x^1;\bar{\theta}^{h}, \cdots,\bar{\theta}^{1})
       = \frac{ \Pr(x^{h+1},x^h,\cdots,x^1 ;\bar{\theta}^{h}, \cdots,\bar{\theta}^{1}) }{\Pr(x^h,\cdots,x^1 ;\bar{\theta}^{h}, \cdots,\bar{\theta}^{1}) }\\
     & = \frac{\Pr(s^{(h+2)m-1}|s^{(h+2)m-2} ;\theta^{(h+2)m-2}) \Pr(s^{(h+2)m-2}|s^{(h+2)m-3} ;\theta^{(h+2)m-3}) \cdots \Pr(s^1|s^0;\theta^0) }
    { \Pr(s^{(h+1)m-1}|s^{(h+1)m-2};\theta^{(h+1)m-2})\Pr(s^{(h+1)m-2}|s^{(h+1)m-3};\theta^{(h+1)m-3}) \cdots \Pr(s^1|s^0;\theta^0 )  } \\
     & = \Pr(s^{(h+2)m-1}|s^{(h+2)m-2} ;\theta^{(h+2)m-2})\cdots \Pr(s^{(h+1)m-1}|s^{(h+1)m-2} ;\theta^{(h+1)m-2}) 
     \\
     & =     \bTr(x^{h+1}|x^h ;\bar{\theta}^{h} ). 
\end{split}
$}
\end{equation}
The \textit{inspection time} $\tau_{IN}^{h,m}=\sum_{{k}'=hm}^{hm+m-1} \tau^{{k}'}$ at inspection stage $h\in  \mathbb{Z}^{0+}$ is a continuous random variable with support $[0,\infty)$  whose PDF $\bar{z}(\cdot|x^h;\bar{\theta}^{h})$ can be computed based on the PDF $z$. 
Based on $\bar{z}$ and $\Pr(w^{hm} |s^{hm},\tau_{IN}^{h,m};\theta^{hm})$ in \eqref{eq:detection probability FE} and \eqref{eq:detection probability RE}, we can compute the probability of security decision $w^{hm}$ at inspection stage $h\in  \mathbb{Z}^{0+}$ given $x^h$ and $\bar{\theta}^{h}$, i.e., 
\begin{equation}
\label{eq:probability of making decision}
\begin{split}
        &\Pr(w^{hm}|x^h,a_m;\bar{\theta}^{h}) 
    =\int_{0}^{\infty} \Pr(w^{hm}, \tau_{IN}^{h,m} |x^h;\bar{\theta}^{h})  d(\tau_{IN}^{h,m}) \\
    & \quad \quad\quad \quad =\int_{0}^{\infty} 
    \Pr(w^{hm} |s^{hm},\tau_{IN}^{h,m};\theta^{hm}) 
    \bar{z}(\tau_{IN}^{h,m}|x^h,\bar{\theta}^{h}) d(\tau_{IN}^{h,m}). 
\end{split}
\end{equation}
Let ${\Pr}(w^{hm}|x^h,a_m;{\theta}^{hm})$ be the shorthand notation for  $\mathbb{E}_{\theta^l\sim [b_{FE},b_{RE}], l\in \{hm+1,\cdots,hm+m-1\}}  \allowbreak [\Pr(w^{hm}|x^h,a_m;\bar{\theta}^{h}) ], \forall \theta^{hm} \in \Theta$. 
We define the probability of the human operator making correct security decisions at inspection stage $h\in \mathbb{Z}^{0+}$ as 
\begin{equation}
\label{eq:pCDbar}
    \hat{p}_{CD}(x^h,a_m):=b_{FE}{\Pr}(w_{FE}|x^h,a_m;{\theta}_{FE}) + b_{RE}{\Pr}(w_{RE}|x^h,a_m;{\theta}_{RE}), \forall x^h\in \mathcal{X},  
\end{equation}
which leads to the  \textit{consolidated severity level} of IDoS attacks in Definition \ref{def:IDoSprob}. 
% Based on $\hat{p}_{CD}(x^h,a_m)$, We define $1-\hat{p}_{CD}(x^h,a_m)$ as the \textit{consolidated severity level} of IDoS attacks under the consolidated state $x^h\in \mathcal{X}$ and AM strategy $a_m\in \mathcal{A}$. 

\begin{definition}[Consolidated Severity Level]
 \label{def:IDoSprob}
% %Consider the sequence of attacks defined in Section \ref{sec:attack} with AM strategy $a_m\in \mathcal{A}$ and a predefined threshold $p_{TH}\in (0,1)$. 
% %If $\hat{p}_{CD}(x^h,a_m)\leq p_{TH}$ for all $x^h\in \mathcal{X}$, then the sequential attack is called an IDoS attack. 
We define $1-\hat{p}_{CD}(x^h,a_m)$ as the \textit{consolidated severity level} of IDoS attacks under the consolidated state $x^h\in \mathcal{X}$ and AM strategy $a_m\in \mathcal{A}$. 
 \end{definition}

% \begin{definition}[Severity Level of IDoS Attack]
% \label{def:IDoSdegree}
% We define the value of $1-\min_{x^h\in \mathcal{X}} \hat{p}_{CD}(x^h,a_m)$ as the severity level of the IDoS attack.  
% \end{definition}

%sum of independent but different exponential distribution is hard to compute. 
%maybe we should consider discrete case for the journal version so that the summation may be easier. For example, consider a Poisson arrival, at each discrete stage, either attack or not attack. 

We denote $c(w^k,s^k;\theta^k)$ as the operator's cost at attack stage $k\in \mathbb{Z}^{0+}$ when the alert's category label is $s^k\in \mathcal{S}$, the attack's type is $\theta^k\in \Theta$, and the security decision is $w^k\in \mathcal{W}$. %and the AM strategy is $a_m\in\mathcal{A}$. %the reward does not depend on a_m
At attack stages where alerts are inconspicuous, i.e., for all $k\neq hm, h\in  \mathbb{Z}^{0+}$, the security decision is  $w_{UN}$ without manual inspection, which incurs an \textit{uncertainty cost}  $c_{UN}>0$.  
At attack stages of highlighted alerts, i.e., for all $k=hm, h\in  \mathbb{Z}^{0+}$, 
the human operator obtains a reward (resp. cost), denoted by $c_{CD}(s^k;\theta^k)<0$ (resp. $c_{ID}(s^k;\theta^k)>0$), for correct (resp. incorrect) security decisions. 
If the human operator remains uncertain about the attack's type after the inspection time $\tau_{IN}^{h,m}$, i.e., $w^{hm}=w_{UN}$, there is the uncertainty cost $c_{UN}$. 
We define the human operator's \textit{consolidated cost} at inspection stage $h\in  \mathbb{Z}^{0+}$ as 
\begin{equation}
\label{eq:consolidated cost}
        % \bar{c}(x^h, a_m; \bar{\theta}^{h}) := \sum_{w^{hm}\in \mathcal{W}}  \Pr(w^{hm}|x^h,a_m;\bar{\theta}^{h}) \sum_{n=0}^{m-1} c(w^{hm+n},s^{hm+n};\theta^{hm+n}),  
        \bar{c}(x^h, a_m; \bar{\theta}^{h}) :=(m-1) c_{UN} + \sum_{w^{hm}\in \mathcal{W}}  \Pr(w^{hm}|x^h,a_m;\bar{\theta}^{h})  c(w^{hm},s^{hm};\theta^{hm}).   
\end{equation}
%where $\sum_{n=0}^{m-1} c(w^{hm+n},s^{hm+n};\theta^{hm+n})=c(w^{hm},s^{hm};\theta^{hm})+(m-1) c_{UN}$. 
%Since $\Pr(w^{hm}|x^h,a_m;\bar{\theta}^{h})$ in \eqref{eq:probability of making decision} depends on $x^h$ and $a_m$, the consolidated cost $\bar{c}$ also does. 

\subsection{Long-Term Risk Measures for IDoS Attacks}
In this section, we define four long-term risk measures whose relations are shown in Fig. \ref{fig:LongTermCosts}. 
The \textit{Cumulative Cost (CC)} and \textit{Expected Cumulative Cost (ECC)} on the left directly follow from the discounted summation of the consolidated cost $\bar{c}$ in \eqref{eq:consolidated cost}. 
Since CC and ECC depend on the consolidated state $x^h$ and  the consolidated type $\bar{\theta}^h$, it is of high dimension and thus difficult to store and compute. 
By taking an expectation over $s^{hm+1}, \cdots, s^{hm+m-1}\in\mathcal{S}$, we reduce the dimension and obtain the  Aggregated Cumulative Cost (ACC) and Expected Aggregated Cumulative Cost (EACC) on the right of the figure. 
The DP representations for CC (resp. ECC) and ACC (resp. EACC) are generally not equivalent. 
We identify the condition under which two DP representations are equivalent in Section \ref{sec:specialcase}. 
The two risk learning schemes are introduced in Section \ref{sec:TD learning}. 
Since the consolidated risk learning is based on ECC, it has to wait for the realization of the consolidated state $x^h:=[s^{hm},\cdots,s^{hm+m-1}]$ to evaluate the inspection performance. 
On the contrary, the EACC-based aggregated risk learning just needs $s^{hm}$ to evaluate the inspection performance, which reduces the dimension of the samples and enables evaluations with no delay.

\begin{figure}[h]
\centering
\includegraphics[width=.8 \textwidth]{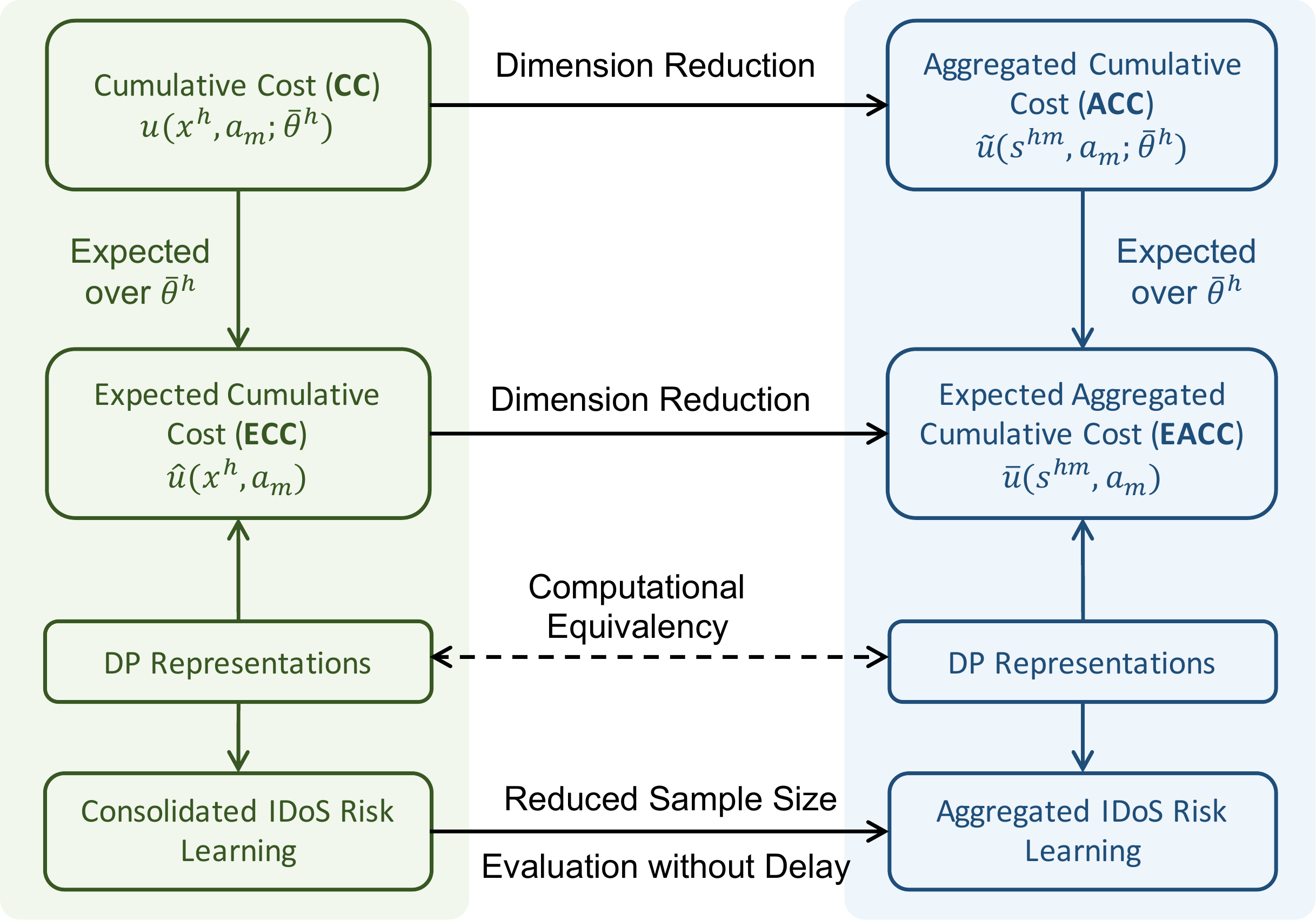}
\caption{ 
Relations among four long-term risk measures, their DP representations, and two risk learning schemes. 
%The ones associated with $x^h$ and $s^{hm}$ are shown in green and blue, respectively.
}
\label{fig:LongTermCosts}
\end{figure}  

\subsubsection{Cumulative Cost and Expected Cumulative Cost}

With discounted factor $\gamma\in (0,1)$, we define the \textit{Cumulative Cost (CC)} under $x^{h_0}$, $\bar{\theta}^{h_0}$, and action $a_m$ as 
$
u(x^{h_0},a_m;\bar{\theta}^{h_0}) := \mathbb{E} [ \sum_{h=h_0}^{\infty}  (\gamma)^h \cdot \bar{c}(x^{h},a_m;\bar{\theta}^{h}) ], 
$
where the expectation is taken over $x^{h_0+n}$ and $\theta^{h_0m+n}$ for all $n\in \{1, 2, \cdots, \infty\}$. 
By Dynamic Programming (DP), we represent $u$ in the following iterative form, i.e., for all $ x^h\in\mathcal{X}$, $\bar{\theta}^{h}\in \bar{\Theta}$, and $ h\in  \mathbb{Z}^{0+}$, 
\begin{equation}
\label{eq:DP}
 u(x^{h},a_m;\bar{\theta}^{h}) = \bar{c}(x^{h},a_m;\bar{\theta}^{h}) + \gamma \sum_{x^{h+1}\in \mathcal{X}} \bTr(x^{h+1}|x^h;\bar{\theta}^{h}) 
     \mathbb{E}_{\bar{\theta}^{h+1}} [
     u(x^{h+1},a_m;\bar{\theta}^{h+1}) ]. 
\end{equation}
%For finite stages, we can compute \eqref{eq:DP} backwardly from $h=H$ to $h=0$. 
%For infinite stages, i.e., $K=\infty$, then $ u= u$ and 
Denote $u^{l}(x^{h},a_m;\bar{\theta}^{h}), l\in \mathbb{Z}^{0+}$, as the estimated value of $u(x^{h},a_m;\bar{\theta}^{h})$ at the $l$-th iteration, we can compute \eqref{eq:DP} by the following value iteration algorithm in Algorithm \ref{algorithm:VI}. 
\begin{algorithm}[h]
\SetAlgoLined
 \textbf{Initialize} a stopping threshold $\epsilon>0$, $l=0$,  and  $u^0(x^{h},a_m;\bar{\theta}^{h})=0, \forall x^h\in \mathcal{X},\bar{\theta}^{h}\in \bar{\Theta}$  \;
   %(\textit{Computation Phase}) 
   \While{ $\max_{ x^h\in \mathcal{X},\bar{\theta}^{h}\in \bar{\Theta} } [u^{l+1}(x^{h},a_m;\bar{\theta}^{h})-u^{l}(x^{h},a_m;\bar{\theta}^{h})]\geq \epsilon$
   }
   {
   \For{$x^h\in \mathcal{X}$ and $\bar{\theta}^{h}\in \bar{\Theta}$}{
Update estimated value $ u^{l+1}(x^{h},a_m;\bar{\theta}^{h}) = \bar{c}(x^{h},a_m;\bar{\theta}^{h}) + \gamma \sum_{x^{h+1}\in \mathcal{X}} \bTr(x^{h+1}|x^h;\bar{\theta}^{h}) 
     \mathbb{E}_{\bar{\theta}^{h+1}} [
     u^l(x^{h+1},a_m;\bar{\theta}^{h+1}) ]$\; 
  }
    
      $l \leftarrow l+1$ \;
  }
  \textbf{Return} $u^{l+1}(x^{h},a_m;\bar{\theta}^{h})$ \; 
 \caption{Value Iteration 
 \label{algorithm:VI}}
\end{algorithm}
It can be shown that $u^{\infty}(x^{h},a_m;\bar{\theta}^{h})$ converges to $u(x^{h},a_m;\bar{\theta}^{h})$ and the following lemma holds \cite{DP}. 
\begin{lemma}[Monotonicity Lemma]
\label{lemma:Monotonicity}
Let $
u'(x^{h_0},a_m;\bar{\theta}^{h_0}) := \mathbb{E} [ \sum_{h=h_0}^{\infty}  (\gamma)^h \cdot \bar{c}'(x^{h},a_m;\bar{\theta}^{h}) ]
$. 
% Consider another set of consolidated cost $\bar{c}'$ and CC $u'$ that statisfy DP
If $\bar{c}(x^{h},a_m;\bar{\theta}^{h})>\bar{c}'(x^{h},a_m;\bar{\theta}^{h}), \forall x^h\in \mathcal{X},\bar{\theta}^{h}\in \bar{\Theta}$, then $ u(x^{h},a_m;\bar{\theta}^{h})>u'(x^{h},a_m;\bar{\theta}^{h})$ for all $x^h\in \mathcal{X},\bar{\theta}^{h}\in \bar{\Theta}$. 
\end{lemma}
We define the \textit{Expected Cumulative Cost (ECC)} as 
$
    \hat{u}(x^{h},a_m):= \mathbb{E}_{\bar{\theta}^{h}} [u(x^{h},a_m;\bar{\theta}^{h})] , \forall x^{h}\in \mathcal{X}
$, and write the DP representation of $\hat{u}$ in \eqref{eq:ECC} by taking expectation over $\bar{\theta}^h$ in \eqref{eq:DP}. 
%Taking expectation over $\bar{\theta}^h$, we can rewrite  \eqref{eq:DP} concerning $\hat{u}$ as 
\begin{equation}
\label{eq:ECC}
    \hat{u}(x^h,a_m)= \mathbb{E}_{ \bar{\theta}^h } [\bar{c}(x^{h},a_m;\bar{\theta}^{h})] + \gamma  \sum_{x^{h+1}\in \mathcal{X}} \mathbb{E}_{ \bar{\theta}^h } [\bTr(x^{h+1}|x^h;\bar{\theta}^{h}) ] 
    \hat{u}(x^{h+1},a_m).  
\end{equation}

\subsubsection{Aggregated Cumulative Cost and Expected Aggregated Cumulative Cost}

We define the \textit{Aggregated Cumulative Cost (ACC)} as  
\begin{equation}
\label{def:ACC}
\begin{split}
        \Tilde{u}(s^{hm},a_m;\bar{\theta}^{h}):=\sum_{s^{hm+1}, \cdots, s^{hm+m-1}\in\mathcal{S}} \bigg[ \Pr(s^{hm+1}, \cdots, s^{hm+m-1}|s^{hm};\bar{\theta}^{h} ) \\
    \cdot u([s^{hm}, \cdots, s^{hm+m-1}], a_m;  \bar{\theta}^{h} ) \bigg], 
\end{split}
\end{equation}
and the \textit{Expected Aggregated Cumulative Cost (EACC)}  as
\begin{equation}
\label{eq:EACC}
    \bar{u}(s^{hm},a_m):= \mathbb{E}_{\theta^l\sim [b_{FE},b_{RE}], l\in \{hm, \cdots,hm+m-1\}} [\Tilde{u}(s^{hm},a_m;\bar{\theta}^{h})] , \forall s^{hm}\in \mathcal{S}. 
\end{equation}

Both ECC $\hat{u}(x^0,a_m)$ and EACC $\bar{u}^0(s^0,a_m)$ evaluate the long-term performance of the AM strategy $a_m\in \mathcal{A}$ on average as defined in Definition \ref{def:IDoS risk}.
However, EACC depends on $s^{hm}$ but not on $s^{hm+1}, \cdots, s^{hm+m-1}$. 

\begin{definition}[Consolidated and Aggregated IDoS risks]
\label{def:IDoS risk}
We define ECC $\hat{u}(x^h,a_m)$ (resp. EACC $\bar{u}(s^{hm},a_m)$) as the consolidated (resp. aggregated) risk of the IDoS attack  under $x^h\in \mathcal{X}$ (resp. $s^{hm}\in \mathcal{S}$) and attention strategy $a_m\in \mathcal{A}$.  
\end{definition}

% We define the Aggregated IDoS attack based on EACC in Definition \ref{def:AggregatedIDoS}. 
% \begin{definition}[Aggregated IDoS Attack]
% \label{def:AggregatedIDoS}
% Consider the sequence of $K$ attacks defined in Section \ref{sec:IDoS model} with $s^0\in \mathcal{S}$ and AM strategy $a_m$. The attack is said to be an Aggregated IDoS attack of level-${u}_{TH}$ if the EACC  $\bar{u}(s^0,a_m)\leq {u}_{TH}$. 
% \end{definition}

% Extension to Bayesian RL means that the distribution of $\theta$ is updated based on Bayesian rules as we observe more and more realization of $\theta$ as time evolve. 
% In classic BCL, only the transition probability depends on  $\theta$. 
% Turn the problem into POMDP. 

\subsection{Inter-Arrival Time with Independent PDF} 
\label{sec:specialcase}
In Section \ref{sec:specialcase}, we consider the special case where PDF $z$ is independent of $s^k$ and $\theta^k$, which reduces the dependency of $\hat{p}_{CD}$ and $\bar{c}$ from $x^h$ to $s^{hm}$ as shown in Lemma \ref{lemma:independent}. 
Moreover, we can obtain DP representations for ACC $\Tilde{u}$ and EACC $\bar{u}$ as shown in Theorem \ref{thm:aggregation state}.  
Value iteration in Algorithm \ref{algorithm:VI} can be revised accordingly to solve these two DP representations.  
%Even if PDF $z$ is independent of $s^k$ and $\theta^k$, CC $u$ still depends on $x^h$ and $\bar{\theta}^{h}$ due to its dependence on $\bTr$. 

\begin{lemma}
\label{lemma:independent}
If PDF $z$ is independent of $s^k$ and $\theta^k$, then  $\hat{p}_{CD}(x^h,a_m)$ in \eqref{eq:pCDbar} can be rewritten as $\hat{p}_{CD}(s^{hm},a_m)$ and the consolidated cost  $\bar{c}(x^{h},a_m ;\bar{\theta}^{h})$  in \eqref{eq:consolidated cost} can be rewritten as $\bar{c}(s^{hm},a_m ;\theta^{hm})$ without loss of generality. 
\end{lemma}
\begin{proof}
If $z$ is independent of $s^k,\theta^k$, then $\bar{z}$ is independent of $x^h,\bar{\theta}^{h}$, and $\Pr(w^{hm}|x^h,a_m;\bar{\theta}^{h})$ in \eqref{eq:probability of making decision} only depends on $s^{hm}$ and $\theta^{hm}$. 
Thus, $\hat{p}_{CD}$ becomes a function of $s^{hm},a_m$, and the consolidated cost $\bar{c}$ in \eqref{eq:consolidated cost} becomes a function of $s^{hm}$, $\theta^{hm}$, and $a_m$. 
\qed
\end{proof}

\begin{theorem}
\label{thm:aggregation state}
If PDF $z$ is independent of $s^k$ and $\theta^k$, then we have the following DP representation in \eqref{eq:DPspecial} for the ACC
\begin{equation}
\label{eq:DPspecial}
\begin{split}
     & \Tilde{u}(s^{hm},a_m;\bar{\theta}^{h}) = \bar{c}(s^{hm},a_m;\theta^{hm}) \\
   &\quad\quad + \gamma \sum_{s^{(h+1)m}\in \mathcal{X}} \Pr(s^{(h+1)m}|s^{hm};\bar{\theta}^{h}) \mathbb{E}_{\bar{\theta}^{h+1}} [\Tilde{u}(s^{(h+1)m},a_m;\bar{\theta}^{h+1})], 
\end{split}
\end{equation}
and the following DP representation in \eqref{eq:DPspecialEACR} for the EACC
\begin{equation}
\label{eq:DPspecialEACR}
\begin{split}
     & \bar{u}(s^{hm},a_m) =\mathbb{E}_{\theta^{hm}} [\bar{c}(s^{hm},a_m;\theta^{hm})] \\
   &\quad\quad + \gamma \sum_{s^{(h+1)m}\in \mathcal{X}} \mathbb{E}_{\bar{\theta}^{h}} [\Pr(s^{(h+1)m}|s^{hm};\bar{\theta}^{h})] \cdot \bar{u}(s^{(h+1)m},a_m), 
\end{split}
\end{equation}
where 
\begin{equation}
\label{eq:expected transition probability}
    \mathbb{E}_{\bar{\theta}^{h}} [\Pr(s^{(h+1)m}|s^{hm};\bar{\theta}^{h})]=
    \sum_{s^{hm+1},\cdots,s^{hm+m-1}\in\mathcal{S}} \prod_{l=hm}^{(h+1)m} \mathbb{E}_{\theta^l\sim [b_{FE},b_{RE}]} [\Tr(s^{l+1}|s^l;\theta^l)]. 
\end{equation}

\end{theorem}
\begin{proof}
First, for all $\bar{\theta}^{h}\in \bar{\Theta}$, we have 
\begin{equation*}
    \sum_{s^{hm+1}, \cdots, s^{hm+m-1}\in\mathcal{S}}  \Pr(s^{hm+1}, \cdots, s^{hm+m-1}|s^{hm};\bar{\theta}^{h} ) \cdot \bar{c}(s^{hm},a_m ;\theta^{hm}) \equiv \bar{c}(s^{hm},a_m ;\theta^{hm}). 
\end{equation*}
Second, since 
\begin{equation*}
    \begin{split}
       & \bTr(x^{h+1}|x^h;\bar{\theta}^{h})=\Pr(s^{(h+1)m},\cdots, s^{(h+2)m-1}|s^{hm+m-1};\bar{\theta}^{h} ) \\
    & \quad \quad =\Pr(s^{(h+1)m+1},\cdots, s^{(h+2)m-1}|s^{(h+1)m};\bar{\theta}^{h} )\Tr(s^{(h+1)m}|s^{hm+m-1};{\theta}^{hm+m-1})
    \end{split}
\end{equation*}
 as shown in \eqref{eq:bTr}, we have 
\begin{equation*}
   \begin{split}
    &    \sum_{s^{hm+1}, \cdots, s^{hm+m-1}\in\mathcal{S}}  \Pr(s^{hm+1}, \cdots, s^{hm+m-1}|s^{hm} ;\bar{\theta}^{h} ) \\ 
       & \quad\quad \quad\quad\quad  \cdot \sum_{x^{h+1}\in \mathcal{X}} \bTr(x^{h+1}|x^h;\bar{\theta}^{h}) \mathbb{E}_{\bar{\theta}^{h+1}} [u(x^{h+1},a_m;\bar{\theta}^{h+1})] \\
       & =  \sum_{s^{hm+1}, \cdots, s^{hm+m-1}\in\mathcal{S}}  \Pr(s^{hm+1}, \cdots, s^{hm+m-1}|s^{hm};\bar{\theta}^{h} ) \\
       & \quad\quad \quad\quad\quad \cdot \sum_{s^{(h+1)m}\in \mathcal{S}}  \Tr(s^{(h+1)m}|s^{hm+m-1};{\theta}^{hm+m-1}) 
          \cdot \mathbb{E}_{\bar{\theta}^{h+1}} [\Tilde{u}(s^{(h+1)m},a_m;\bar{\theta}^{h+1})]. 
   \end{split}
\end{equation*}
Based on the Markov property, we have
\begin{equation*}
\begin{split}
        & \sum_{s^{hm+1}, \cdots, s^{hm+m-1}\in\mathcal{S}}  \Pr(s^{hm+1}, \cdots, s^{hm+m-1}|s^{hm} ;\bar{\theta}^{h}) \Tr(s^{(h+1)m}|s^{hm+m-1};{\theta}^{hm+m-1}) \\
     & \quad\quad\quad \quad\quad\quad = \Pr(s^{(h+1)m}|s^{hm};\bar{\theta}^{h}). 
\end{split}
\end{equation*}
Therefore, we obtain \eqref{eq:DPspecial} by plugging \eqref{eq:DP} into the definition of ACC in \eqref{def:ACC}. 
We obtain \eqref{eq:DPspecialEACR} by taking expectation over $\bar{\theta}^h$ and using the definition of EACC in \eqref{eq:EACC}. 
Based on \eqref{eq:expected transition probability}, we can compute $\mathbb{E}_{\bar{\theta}^{h}} [\Pr(s^{(h+1)m}|s^{hm};\bar{\theta}^{h})]$ directly from  the transition probability $\Tr$ by the forward Kolmogorov equation. 
\qed
\end{proof}

\begin{remark}[\textbf{Computational Equivalency}] 
To compute $\Tilde{u}$, we generally need to first compute $u$ via $\eqref{eq:DP}$ and then take expectation over $s^{hm+1}, \cdots,s^{hm+m-1}$. 
This computation is of high temporal and spatial complexity as $u$ depends on $x^h$. 
However, for the special case where $z$ is independent of $s^k$ and $\theta^k$, we can compute $\Tilde{u}$ directly based on \eqref{eq:DPspecial} and reduce the computational complexity. 
Thus, Theorem \ref{thm:aggregation state} establishes a computational equivalency between the two DP representations in \eqref{eq:DP} and \eqref{eq:DPspecial}, which contributes to a lightweight computation scheme. 
Analogously, we also establish a computational equivalency between the two DP representations in \eqref{eq:ECC} and \eqref{eq:DPspecialEACR} by taking expectations of \eqref{eq:DP} and \eqref{eq:DPspecial} with respect to $\bar{\theta}^h$. 
\end{remark}

\subsection{Data-Driven Assessment}
\label{sec:TD learning}
In practice, we do not know the parameters of the SMP model, including the transition probability $\Tr$, the PDF $z$, the threshold set $\mathcal{N}(s^k,\theta^k)$, and the set of probability of making correct (resp. incorrect) decisions $\bar{p}_{CD}^n$ (resp. $\bar{p}_{ID}^n$), $n\in \{0,1\cdots,N\}$.  
Therefore, we use Temporal-Difference (TD) learning \cite{DP} to evaluate the performance of the AM strategy $a_m\in \mathcal{A}$ based on the inspection results in real-time. 
%In Section \ref{sec:TD learning}, we consider $K=\infty$  (i.e., $H=\infty$) and omit the superscript of the stage index. 
\subsubsection{Consolidated IDoS Risk Learning}
Letting $v^h(x^h,a_m)$ be the estimated value of $\hat{u}(x^h,a_m)$ at the inspection stage $h\in  \mathbb{Z}^{0+}$, we have the following recursive update in real-time as shown in \eqref{eq:TDgeneral}. 
\begin{equation}
\label{eq:TDgeneral}
    v^{h+1}(\hat{x}^h,a_m)=(1-\alpha^h(\hat{x}^h) ) v^{h}(\hat{x}^h,a_m)+ \alpha^h(\hat{x}^h) (\hat{c}^h+\gamma v^{h}(\hat{x}^{h+1},a_m) ), 
\end{equation}
where $\hat{x}^h$ (resp. $\hat{x}^{h+1}$) is the observed state value at the current inspection stage $h$ (resp. the next inspection stage $h+1$), $\alpha^h(\hat{x}^h)\in (0,1)$ is the learning rate, and $\hat{c}^h$ is the observed cost at stage $h\in  \mathbb{Z}^{0+}$. 
To guarantee that $v^{\infty}$ convergences to $\hat{u}$, we require $\sum_{h=0}^{\infty}\alpha^h({x}^h)=\infty$ and $\sum_{h=0}^{\infty} (\alpha^h({x}^h))^2<\infty$ for all $x^h\in\mathcal{X}$. 
%Based on the value of $v^h(x^h,a_m)$ and a pre-defined threshold $u_{TH}$, we can identify the IDoS attack based on Definition \ref{def:IDoS}. 

\subsubsection{Aggregated IDoS Risk Learning}
For the special case where PDF $z$ is independent of $s^k$ and $\theta^k$, we can use TD learning to directly estimate EACC $\bar{u}(s^{hm},a_m)$ based on \eqref{eq:DPspecialEACR}. 
 Letting $\bar{v}^h(x^h,a_m)$ be the estimated value of $\bar{u}(s^{hm},a_m)$ at the inspection stage $h\in  \mathbb{Z}^{0+}$, we have the following recursive update in real-time as shown in \eqref{eq:TDspecial}. 
\begin{equation}
\label{eq:TDspecial}
    \bar{v}^{h+1}(\hat{s}^{hm},a_m)=(1-\bar{\alpha}^h(\hat{s}^{hm}) ) v^{h}(\hat{s}^{hm},a_m)+ \bar{\alpha}^h(\hat{s}^{hm}) (\hat{c}^h+\gamma \bar{v}^{h}(\hat{s}^{(h+1)m},a_m) ), 
\end{equation}
where $\hat{s}^{hm}$ (resp. $\hat{s}^{(h+1)m}$) is the observed state value at the current inspection stage $h$ (resp. the next inspection stage $h+1$), $\bar{\alpha}^h(\hat{s}^{hm})\in (0,1)$ is the learning rate, and $\hat{c}^h$ is the observed cost at stage $h\in  \mathbb{Z}^{0+}$. 
To guarantee that $\bar{v}^{\infty}$ convergences to $\bar{u}$, we require $\sum_{h=0}^{\infty}\bar{\alpha}^h({s}^{hm})=\infty$ and $\sum_{h=0}^{\infty} (\bar{\alpha}^h({s}^{hm}))^2<\infty$ for all $s^{hm}\in\mathcal{S}$. 
%Based on the value of $\bar{v}^h(s^{hm},a_m)$ and a pre-defined threshold $u_{TH}$, we can identify the aggregated IDoS attack based on Definition \ref{def:IDoS}. 

\section{Numerical Experiments and Analysis}
\label{sec:casestudy}
We provide a numerical case study in this section to corroborate the results. Let the set of category label $\mathcal{S}=\{s_{AL},s_{NL},s_{PL}\}$ be the location of the attacks where $s_{AL}$, $s_{NL}$, and $s_{PL}$ represent the application layer, network layer, and physical layer, respectively. 
%, based on a combination of the levels of lethality, criticality, defensibility, and resiliency in Section \ref{sec:attack}. 
%For example, `a new virus with the potential to spread quickly' is of medium risk and `Mission-critical application failures' is of high risk.
%https://www.cisecurity.org/cybersecurity-threats/alert-level/
We consider the special case where $\tau^k, \forall k\in \mathbb{Z}^{0+}$, is an exponential random variable with a constant rate $\beta>0$, i.e., $z(\tau |s^k, \theta^k)= \beta e^{{-\beta \tau}}, \forall s^k\in \mathcal{S}, \theta^k\in\Theta, \tau\in [0,\infty)$. 
Fig. \ref{fig:statedynamics} illustrates an exemplary sequential attack where the vertical dashed lines represent the attack stages $k\in \mathbb{Z}^{0+}$. 
The length of the rectangles between the $k$-th and $(k+1)$-th vertical dash lines represents the $k$-th attack's duration $\tau^k$. 
The height of each square distinguishes the attack's type; i.e., tall and short rectangles represent feints and real attacks, respectively. 

\begin{figure}[h]
\centering
\includegraphics[width=1 \textwidth]{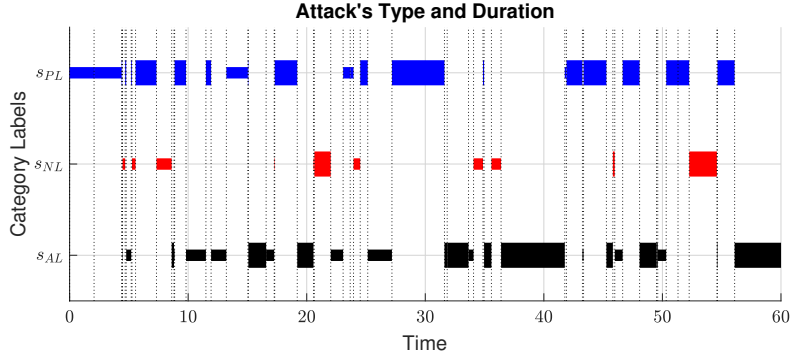}
\caption{ 
The sequential arrival of feints and real attacks with different category labels.  
}
\label{fig:statedynamics}
\end{figure}

The inspection time $\tau_{IN}^{h,m}$, as the summation of $m$ i.i.d. exponential random variables, is an Erlang distribution with shape $m$ and and rate $\beta>0$, i.e., 
$ \bar{z}(\tau|x^h;\bar{\theta}^{h})=\frac {\beta^{m} \tau^{m-1}e^{-\beta \tau}}{(m-1)!},\forall x^h\in \mathcal{X},\forall \bar{\theta}^{h},\tau\in [0,\infty) $. 
Consider a single threshold $N=1$ and $\mathcal{N}=\{\bar{\tau}_0(s^k,\theta^k),\bar{\tau}_N(s^k,\theta^k) \}$. 
%We further assume that given sufficient inspection time, the human operator does not make incorrect security decision, i.e., $\bar{p}_{ID}^n(\theta^k, s^k)=0, \forall  s^k\in \mathcal{S}, \theta^k\in\Theta, \forall n\in \{0,1,\cdots,N\}$.  
Then, $\Pr(w^{hm}|x^h,a_m;\bar{\theta}^{h})$ in \eqref{eq:probability of making decision} has the following closed form in \eqref{eq:closedformProb} for correct decisions, i.e.,  $\theta^{hm}=\theta_{FE},w^{hm}=w_{FE}$ or $\theta^{hm}=\theta_{RE},w^{hm}=w_{RE}$. 
% Now, we compute the pdf by 
% \begin{equation}
% \begin{split}
%         \Pr(W=w_{RE})=\int_{\bar{\tau}(s^k)}^{\infty} \Pr(W=w_{RE},\tau) d\tau = 
%     \int_{\bar{\tau}(s^k)}^{\infty} \Pr(W=w_{RE}|\tau)f(\tau) d\tau \\
% =   \int_{\bar{\tau}(s^k)}^{\infty}\bigg[ 2/(1+e^{-(\tau_{IN}^k-\bar{\tau}(s^k))})-1 \bigg] \frac{\beta^m\tau^{m-1}e^{\beta \tau}}{(m-1)!} d\tau. 
% \end{split}
% \end{equation}
% we do not have closed form as shown in Mathematica (it contains zeta function). 
\begin{equation}
    \begin{split}
    \label{eq:closedformProb}
       \Pr(w^{hm}|s^{hm},a_m;{\theta}^{hm}) &= \int_{\bar{\tau}_N(s^{hm},\theta^{hm})}^{\infty} \bar{p}_{CD}^N( s^{hm},\theta^{hm}) \frac{\beta^m\tau^{m-1}e^{-\beta \tau}}{(m-1)!} d\tau \\
&=\bar{p}_{CD}^N( s^{hm},\theta^{hm})  (1-CDF( \bar{\tau}_N(s^{hm},\theta^{hm}) )),
    \end{split}
\end{equation}
where the Cumulative Distribution Function (CDF) of the random variable $\tau_{IN}^{h,m}$ is
\begin{equation}
\label{eq:cdf}
    CDF( \bar{\tau}_N(s^{hm},\theta^{hm}) ) = 1- \sum_{n=0}^{m-1} \frac{1}{n!} e^{-\beta \bar{\tau}_N(s^{hm},\theta^{hm}) } (\beta \bar{\tau}_N(s^{hm},\theta^{hm}) )^n. 
\end{equation}

%Then, $\hat{p}_{CD}(x^h,a_m)$
%Then, $ \Pr(W=w_{UN}|\theta_{RE},s^k) =1- \Pr(W=w_{RE}|\theta_{RE},s^k)$ and $ \Pr(W=w_{FE}|\theta_{RE},s^k)=0$. 

% The consolidated cost in \eqref{eq:consolidated cost} becomes a function of $s^{hm}$, $a_m$, and $\theta^{hm}$, i.e., 
% \begin{equation*}
%  \begin{split}
%      \bar{c}(s^{hm}, a_m; {\theta}^{hm})  := (m-1)c_{UN} + \sum_{w^{hm}\in \mathcal{W}} \Pr(w^{hm}|s^{hm},a_m;{\theta}^{hm}) c(w^{hm},s^{hm};\theta^{hm}). 
%  \end{split}
% \end{equation*}

\subsection{Value Iteration and TD Learning}
Since PDF $z$ is independent of $s^k$ and $\theta^k$, we can compute EACC in \eqref{eq:DPspecialEACR} by value iteration. 
As shown in Fig. \ref{fig:valueiteration}, the estimated values of EACC under three different category labels, i.e., $\bar{u}(s_{AL},a_m)$, $\bar{u}(s_{NL},a_m)$, and $\bar{u}(s_{PL},a_m)$ in black, red, and blue, respectively, all converge within $40$ iterations. 
When the exact model is unknown, we use TD learning in \eqref{eq:TDspecial} to estimate EACC $\bar{u}(s^{hm,a_m})$. 
In particular, we choose $\bar{\alpha}^h({s}^{hm})=\frac{k_c}{k_{TI}(s^{hm})-1+k_c}$ as the learning rate where $k_c\in (0,\infty)$ is a constant parameter and $k_{TI}(s^{hm})\in \mathbb{Z}^{0+}$ is the number of visits to $s^{hm}\in\mathcal{S}$ up to stage $h\in  \mathbb{Z}^{0+}$. 
We illustrate the convergence of TD learning in Fig. \ref{fig:TDkc10} with $k_c=6$.  
Since the number of visits to  $s_{AL}$, $s_{NL}$, and $s_{PL}$ depends on the transition probability $\bTr$, the learning stages for three category labels are of different lengths. 

\begin{figure}[h]
    \centering % <--a
\begin{subfigure}{0.5\textwidth}
  \includegraphics[width=\linewidth]{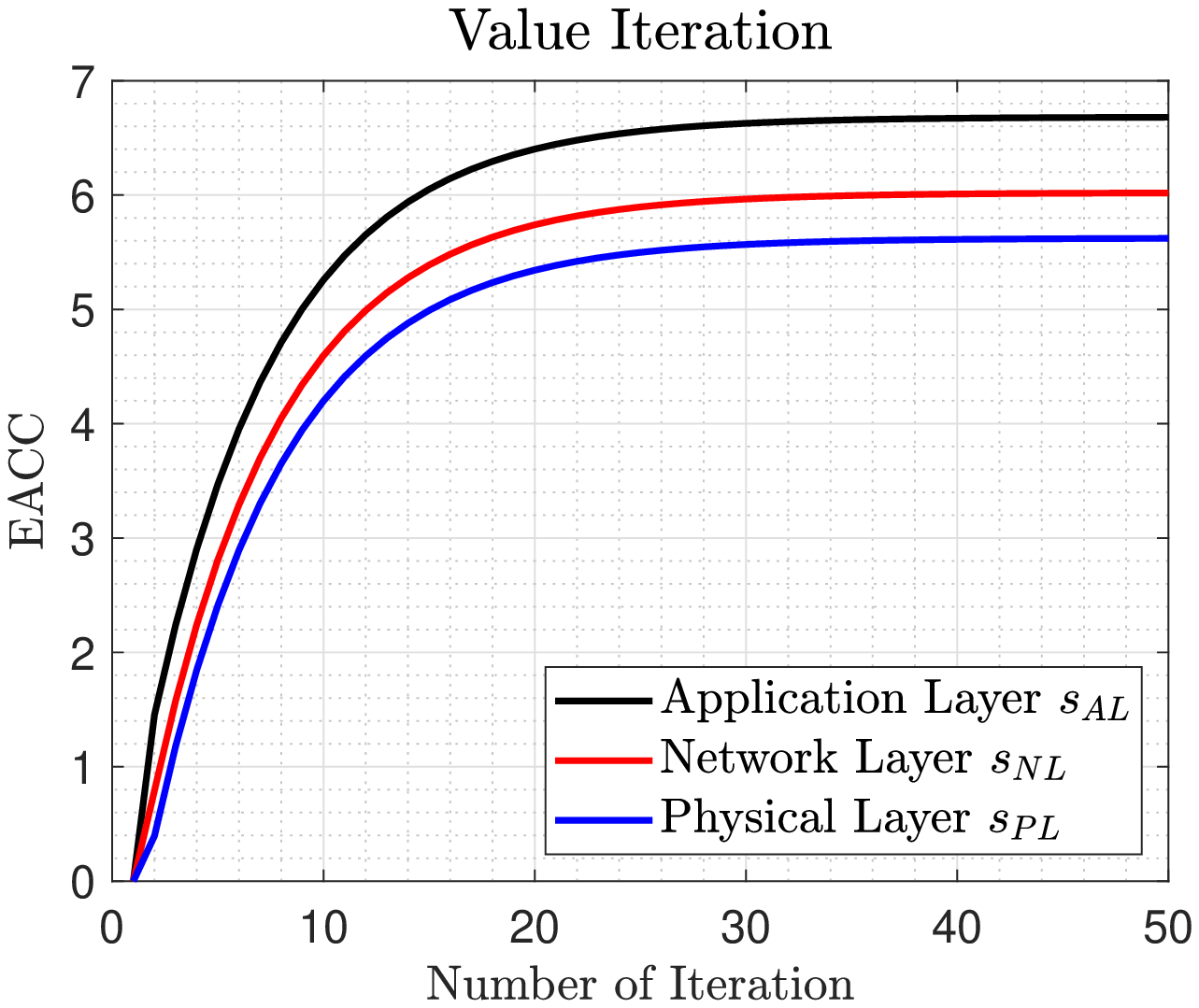}
  \caption{\label{fig:valueiteration} 
  Theoretical value by  value iteration. 
   }
\end{subfigure}\hfil 
\begin{subfigure}{0.5\textwidth} % <--b
  \includegraphics[width=\linewidth]{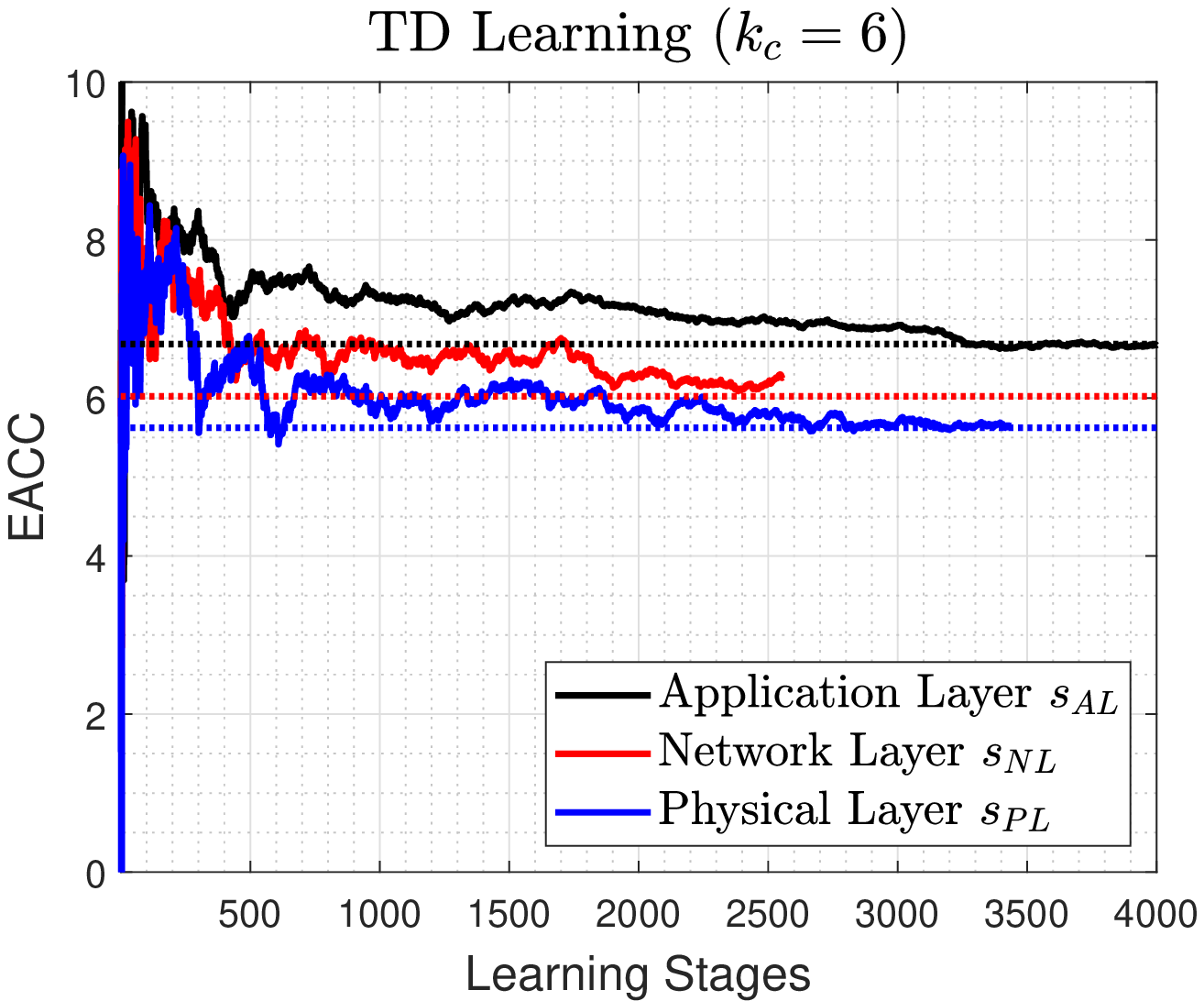}
  \caption{\label{fig:TDkc10} 
  Simulated value by  TD learning. 
 }
\end{subfigure}\hfil 
\caption{
Computation and learning of EACC. 
}
\label{fig:correct}
\end{figure}

\begin{figure}[h]
    \centering % <--a
\begin{subfigure}{0.5\textwidth}
  \includegraphics[width=\linewidth]{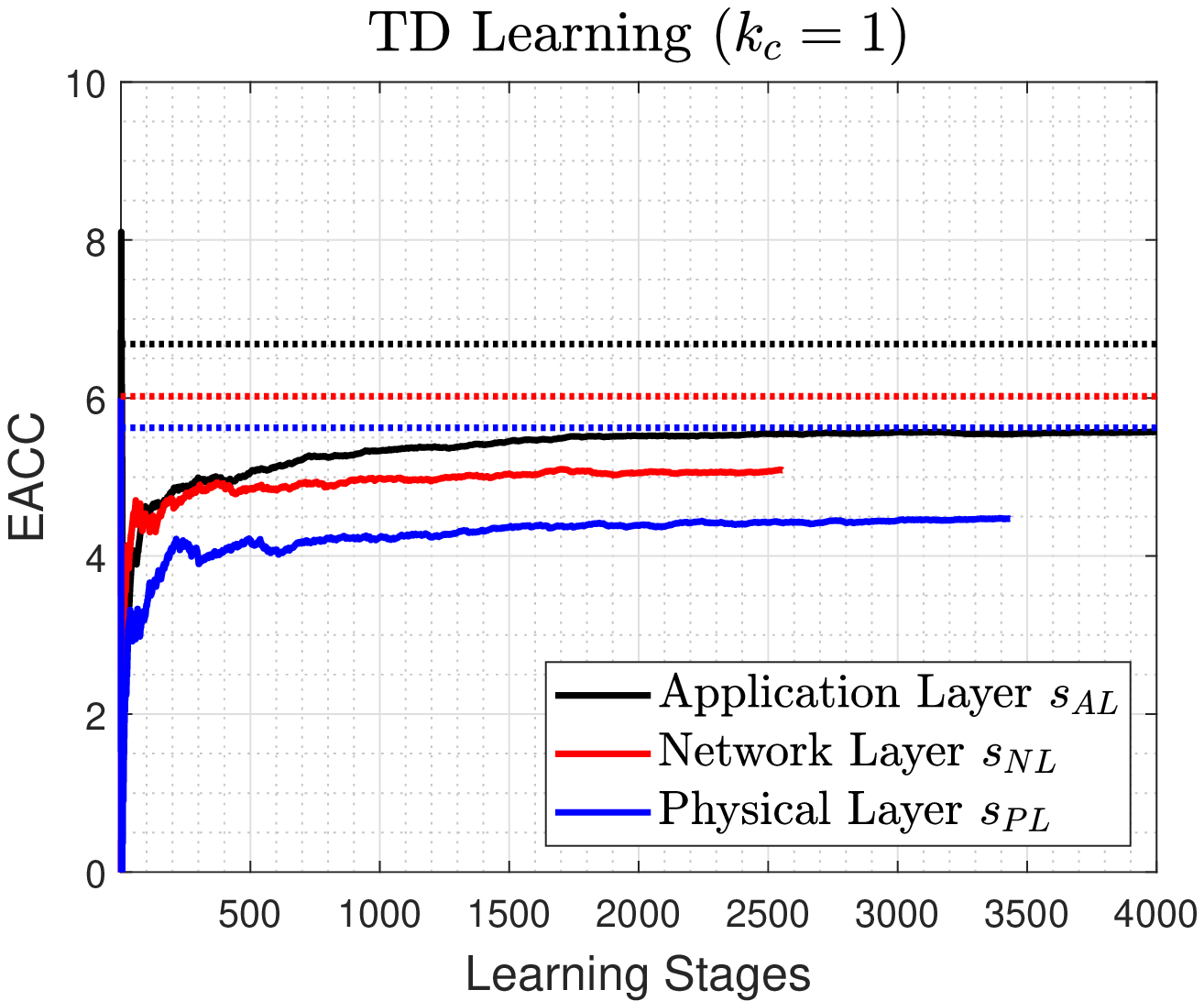}
  \caption{\label{fig:TDkc5} 
  Learning rate decreases too fast $k_c=1$. 
   }
\end{subfigure}\hfil 
\begin{subfigure}{0.5\textwidth} % <--b
  \includegraphics[width=\linewidth]{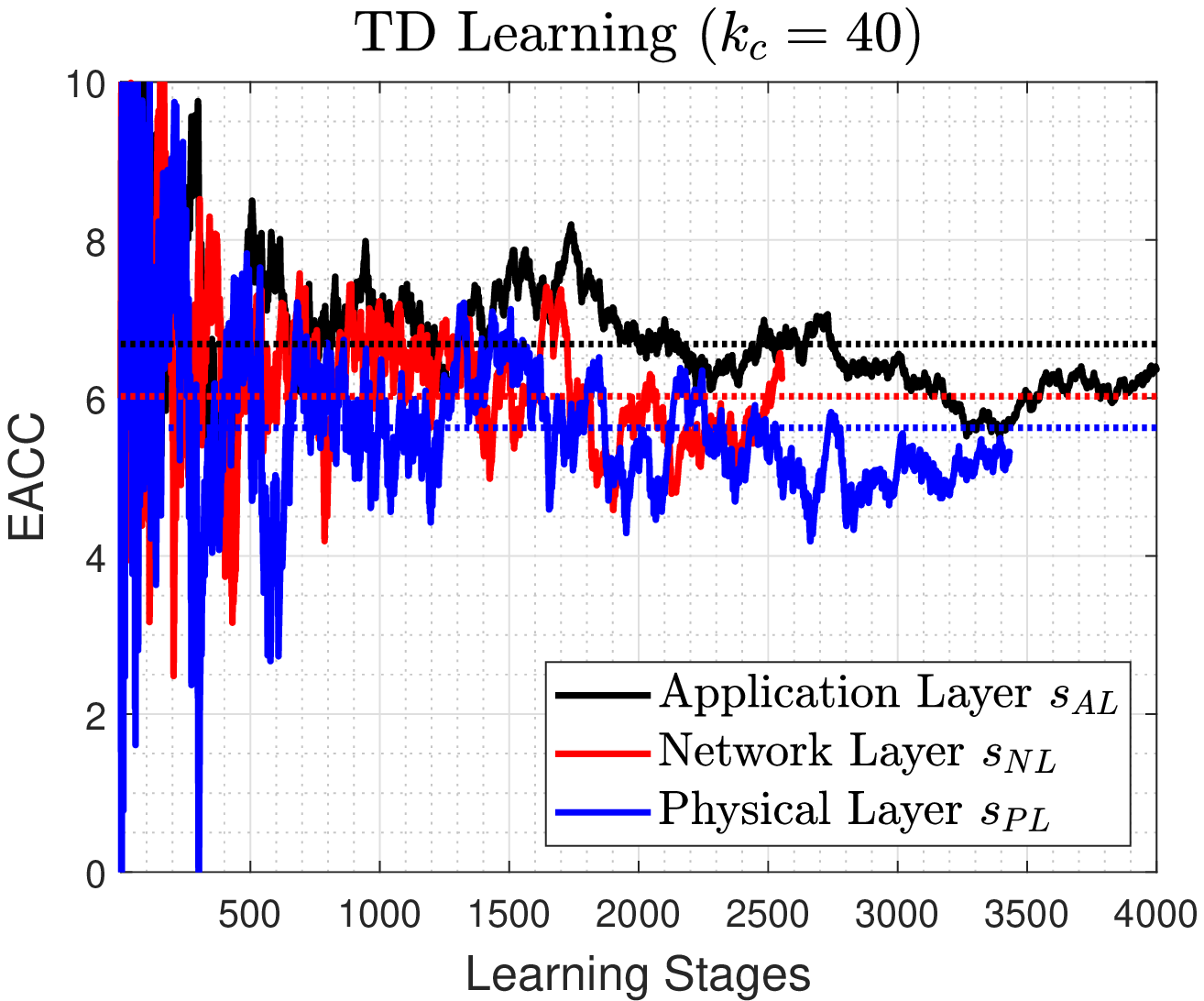}
  \caption{\label{fig:TDkc80} 
Learning rate decreases too slow $k_c=40$.
 }
\end{subfigure}\hfil 
\caption{
Improper values of $k_c$ lead to unsatisfactory learning performances in finite steps. 
}
\label{fig:TDwrong}
\end{figure}

If $k_c$ is too small as shown in Fig. \ref{fig:TDkc5}, the learning rate
decreases so fast that new observed samples hardly update the estimated value. Then, it takes  longer learning stages to learn the correct value. 
On the contrary, if  $k_c$ is too large as shown in Fig. \ref{fig:TDkc80},  the learning rate decreases so slow that new samples contribute significantly to the current estimated value, which causes a large variation and a slow convergence.

\subsection{Severity Level and Aggregated Risk without Attention Management} 
\label{sec:numericalm=1}
%In this section, we fix m=1, and inspect how the frequency beta affect the performance. That is when supply bigger or smaller than the demand. 
When there are no AM strategies, i.e., $m=1$, the human operator switches attention whenever a new attack arrives. 
Then, 
\eqref{eq:cdf} can be simplified as  
$
  CDF( \bar{\tau}_N(s^{hm},\theta^{hm}) ) = 
1-  e^{-\beta \bar{\tau}_N(s^{hm},\theta^{hm}) } , \forall s^{hm}\in \mathcal{S}, \theta^{hm}\in \Theta
$, which is an exponential function of the product of the rate $\beta>0$ and the threshold  $ \bar{\tau}_N(s^{hm},\theta^{hm})>0$. 
Thus, $\hat{p}_{CD}(x^h,a_m)$ in \eqref{eq:pCDbar} decreases monotonously as the value of the product $\beta \bar{\tau}_N(s^{hm},\theta^{hm})$ increases. 
Based on Lemma \ref{lemma:independent}, we can write the consolidated severity level as $1-\hat{p}_{CD}(s^{hm},a_m)$ without loss of generality. 
% Similarly, the values of $\mathbb{E}_{\theta^{hm}} [\bar{c}(s^{hm},a_m;\theta^{hm})]$ in  \eqref{eq:rbarspecial} and $\bar{u}(s^0;a_m)$ in \eqref{eq:DPspecialEACR} both
% increase exponentially as the value of the product $\beta \bar{\tau}_N(s^{hm},\theta^{hm})$ increases. 
% We summarize the above results in Proposition \ref{prop:exp_product}. 
% \begin{proposition}
% \label{prop:exp_product}
% If $m=1$ and %$\tau^k, \forall k\in \mathbb{Z}^{0+}$, is an exponential random variable with PDF 
% $z(\tau |s^k, \theta^k)= \beta e^{{-\beta \tau}}, \forall s^k\in \mathcal{S}, \theta^k\in\Theta, \tau\in [0,\infty)$, then both the severity level and the aggregated risk of IDoS attacks increase exponentially with the value of $\beta \bar{\tau}_N(s^{hm},\theta^{hm})$. 
% \end{proposition}
% \begin{remark}[\textbf{Supply and Demand of Attention Resources}]
%  On the one hand, as $\beta$ increases, the sequential attacks arrive with a higher frequency on average, which results in a higher demand of attention resources from the human operator. 
%  On the other hand, as $\bar{\tau}_N$ increases, the human operator's \textit{inspection efficiency} decreases. Then, the operator requires a longer inspection time to determine the attack's type, which results in a lower supply of attention resources. 
%  Without a proper AM strategy (i.e., $m=1$), Proposition \ref{prop:exp_product} quantifies the product of the supply and demand as the determinant of the severity level and the aggregated risk of IDoS attacks. 
% \end{remark}
\begin{figure}[thb]
    \centering % <--a
\begin{subfigure}{0.5\textwidth}
  \includegraphics[width=\linewidth]{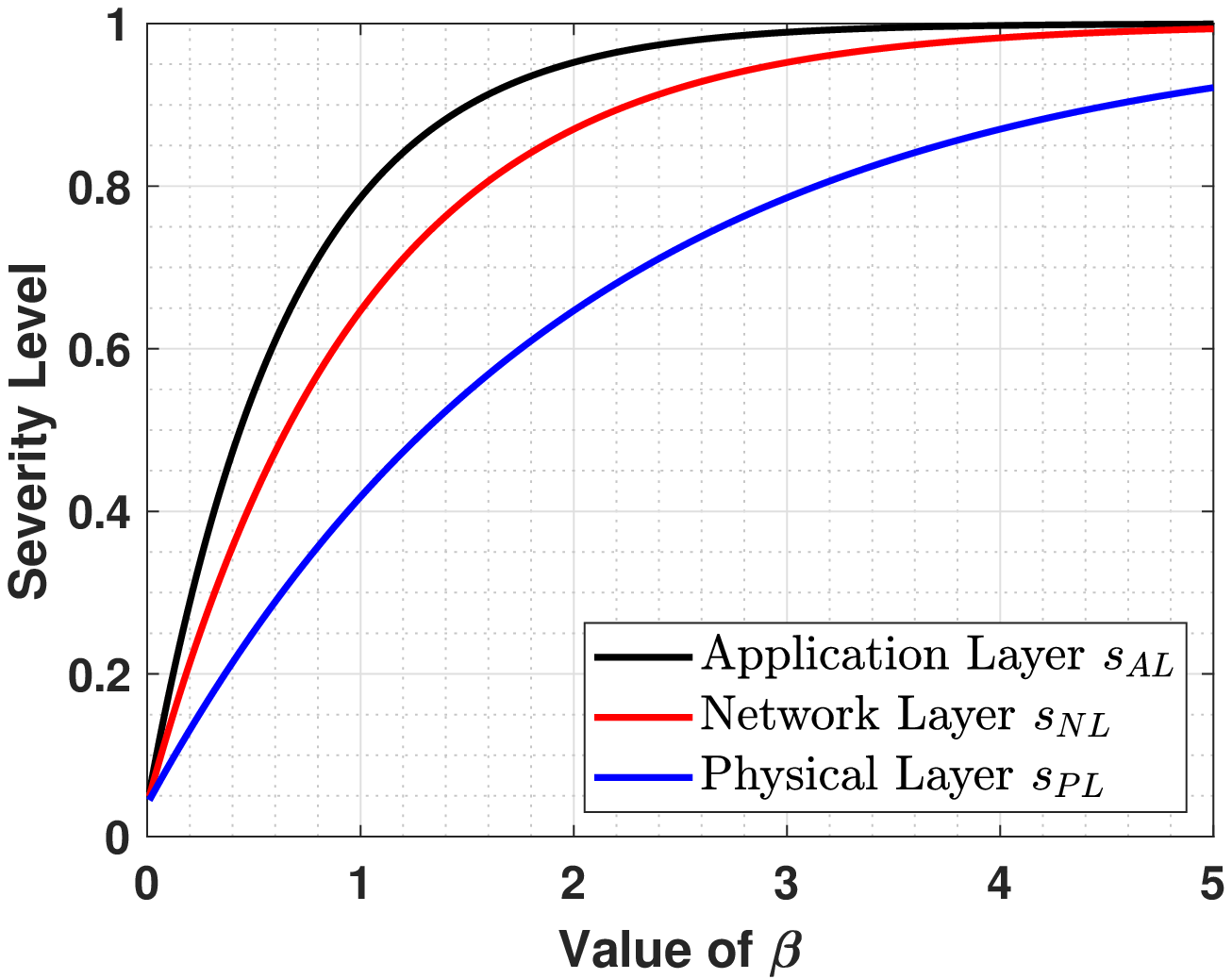}
  \caption{\label{fig:Severitylevel} 
 Severity level vs. $\beta$. 
   }
\end{subfigure}\hfil 
\begin{subfigure}{0.5\textwidth} % <--b
  \includegraphics[width=\linewidth]{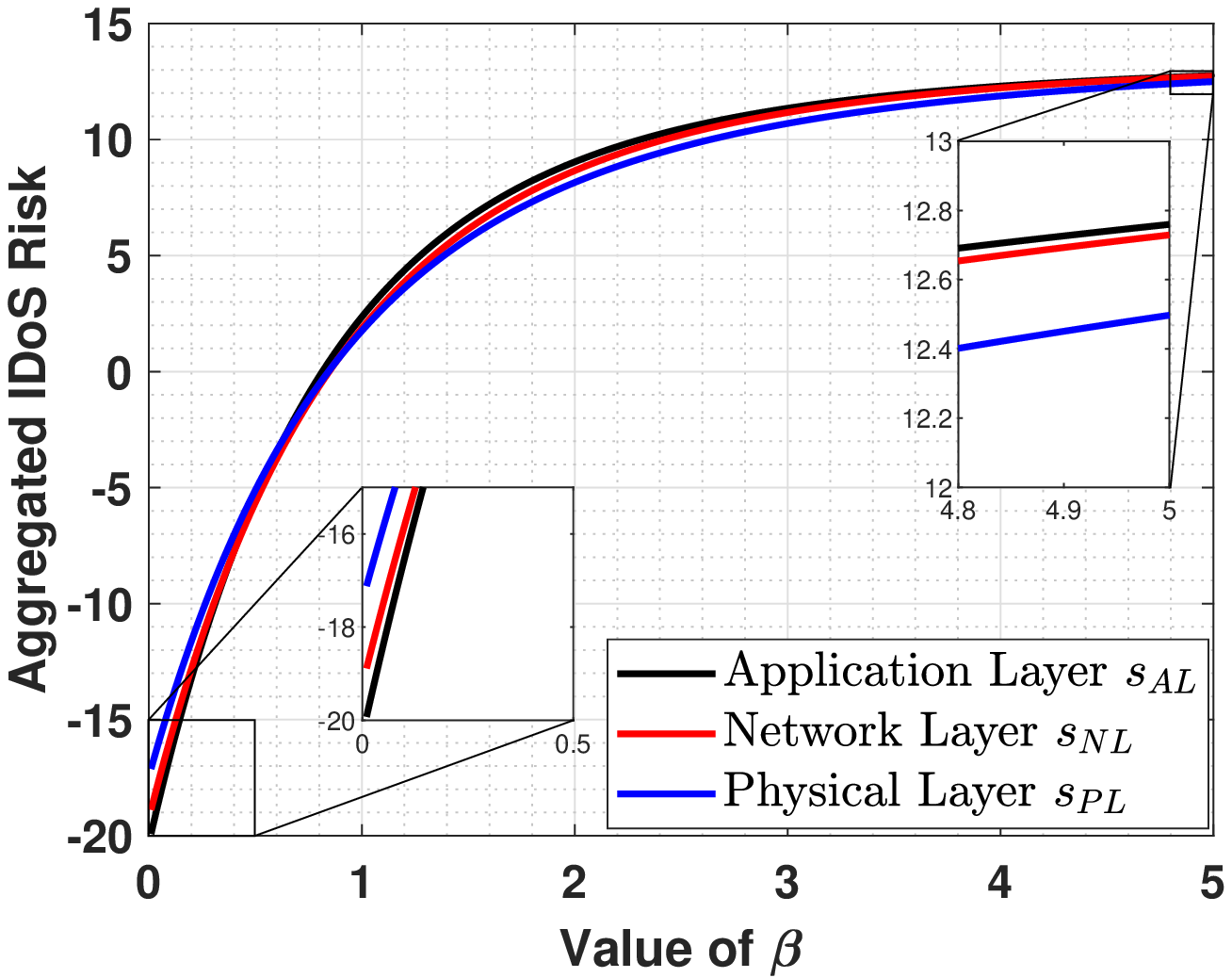}
  \caption{\label{fig:Aggregatedrisk} 
Aggregated risk vs. $\beta$.
 }
\end{subfigure}\hfil 
\caption{
 Severity level and aggregated risk of IDoS attacks under $s_{AL}$, $s_{NL}$, and $s_{PL}$ in black, red, and blue, respectively. 
 The insert boxes magnify the selected areas. 
}
\label{fig:withoutAM}
\end{figure}
%Fig. \ref{fig:withoutAM} corroborates the results in Proposition \ref{prop:exp_product}.  
Let $\bar{\tau}_N(s_{AL},\theta^{hm})\geq \bar{\tau}_N(s_{NL},\theta^{hm})\geq \bar{\tau}_N(s_{PL},\theta^{hm})$, we plot the  severity level, i.e., $1-\hat{p}_{CD}(s^{hm},a_m)$, for different values of rate $\beta\in (0,5)$ in Fig. \ref{fig:Severitylevel}. 
We illustrate the aggregated IDoS risk versus $\beta\in (0,5)$ in Fig. \ref{fig:Aggregatedrisk}. 
As magnified by two insert boxes, the aggregated IDoS risk under $s_{AL}$, $s_{NL}$, and $s_{PL}$ can change orders for different $\beta$.

\subsection{Severity Level and Aggregated Risk with Attention Management}
\label{sec:numericalm>1}
We illustrate how different AM strategies affect the severity level and the aggregated risk of IDoS attacks 
\begin{figure}[h]
    \centering % <--a
\begin{subfigure}{0.5\textwidth}
  \includegraphics[width=\linewidth]{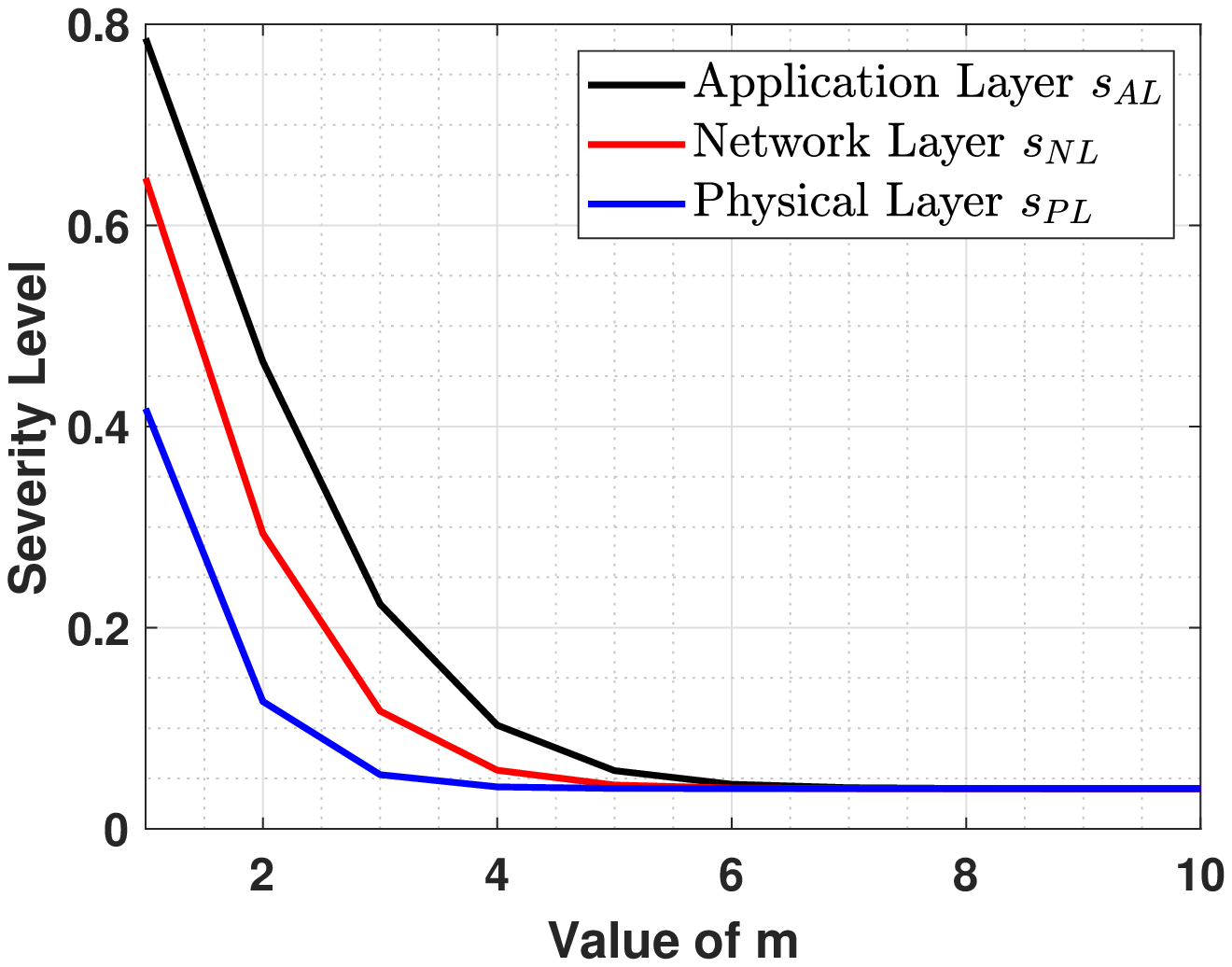}
  \caption{\label{fig:SeveritylevelVSmbeta1} 
 Small arrival rate $\beta=1$. 
   }
\end{subfigure}\hfil 
\begin{subfigure}{0.5\textwidth} % <--b
  \includegraphics[width=\linewidth]{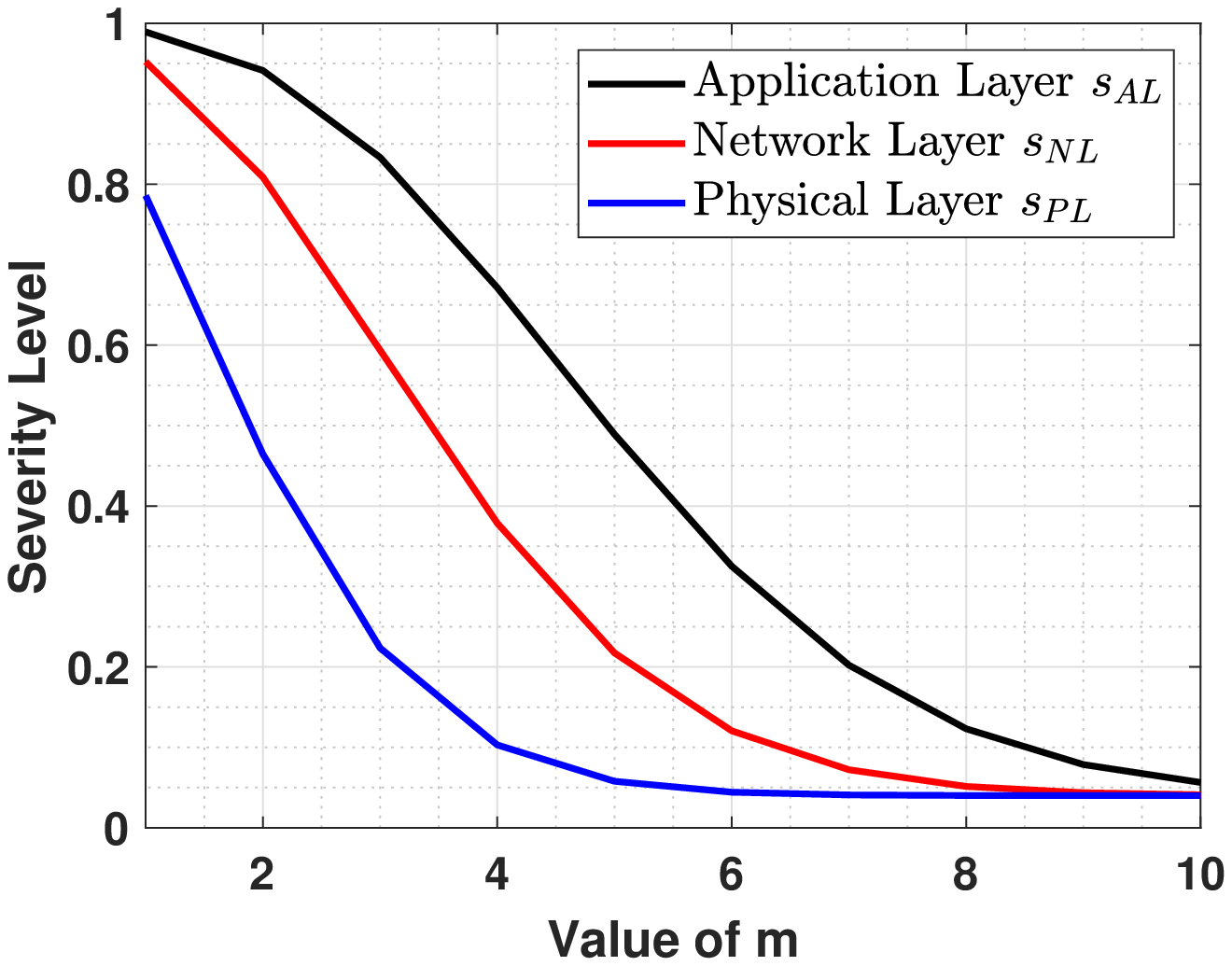}
  \caption{\label{fig:SeveritylevelVSmbeta3} 
Large arrival rate $\beta=3$. 
 }
\end{subfigure}\hfil 
\caption{
Severity levels of IDoS attacks under $s_{AL}$, $s_{NL}$, and $s_{PL}$ in black, red, and blue. 
}
\label{fig:withAM1}
\end{figure}
\begin{figure}[h]
    \centering % <--a
\begin{subfigure}{0.5\textwidth}
  \includegraphics[width=\linewidth]{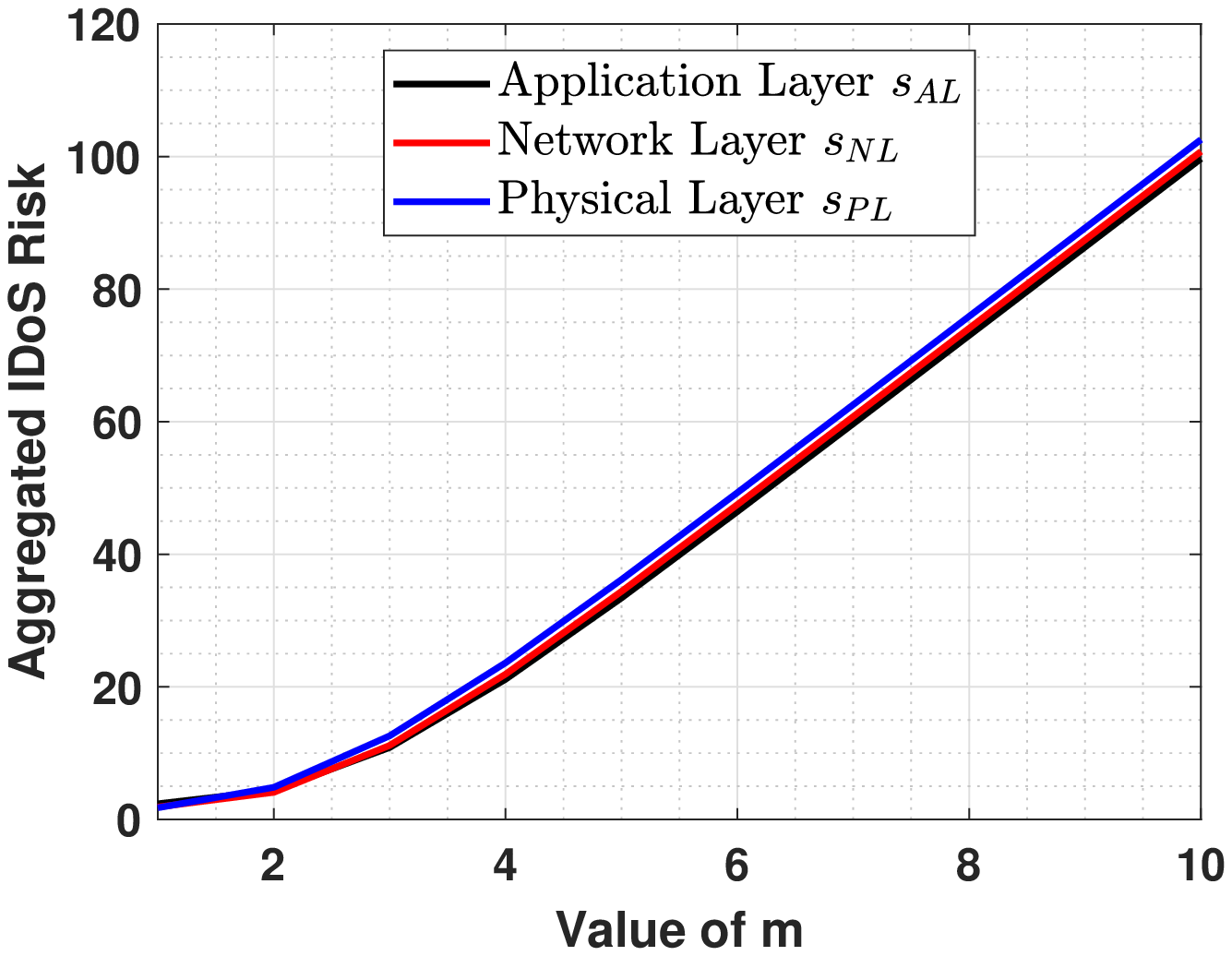}
  \caption{\label{fig:AggregatedriskVSmcUN20} 
High cost $c_{UN}(s^{hm}, \theta^k)=20,\forall s^{hm}, \theta^{hm}$. 
   }
\end{subfigure}\hfil 
\begin{subfigure}{0.5\textwidth} % <--b
  \includegraphics[width=\linewidth]{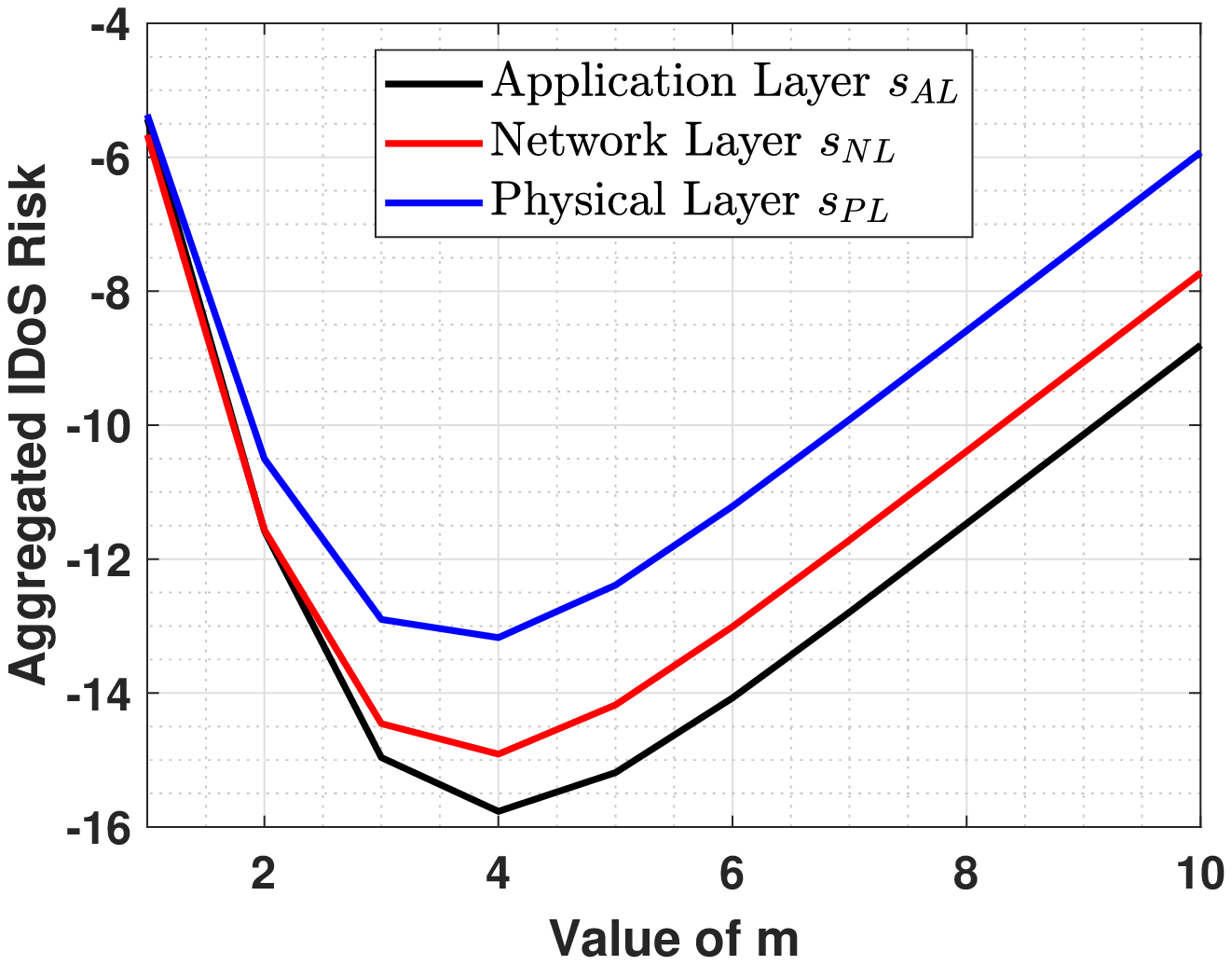}
  \caption{\label{fig:AggregatedriskVSmcUN02}
Low cost $c_{UN}(s^{hm}, \theta^k)=0.2,\forall s^{hm}, \theta^{hm}$. 
 }
\end{subfigure}\hfil 
\caption{
Aggregated IDoS risks under $s_{AL}$, $s_{NL}$, and $s_{PL}$ in black, red, and blue. 
}
\label{fig:withAM2}
\end{figure}
%We corroborate Lemma \ref{lemma:SLwithAM} and Proposition  \ref{prop:ARiskwithAM} 
in Fig. \ref{fig:withAM1} and Fig. \ref{fig:withAM2}, respectively, where $\bar{p}_{CD}^N(s^{hm}, \theta_{FE})=1$ and $\bar{p}_{CD}^N(s^{hm}, \theta_{RE})=0.9$ for all $s^{hm}\in \mathcal{S}$, and $b_{FE}=0.6$. 
As shown in Fig. \ref{fig:withAM1}, the severity level strictly decreases to $0.04$ as $m$ increases regardless of different values of $\beta$. 
We choose a small arrival rate $\beta=1$ in Fig. \ref{fig:SeveritylevelVSmbeta1} and a large rate  $\beta=3$ in Fig. \ref{fig:SeveritylevelVSmbeta3}. 
For a given $m\in \mathbb{Z}^+$, a larger arrival rate results in a higher severity level, and more alerts need to be made inconspicuous to reduce the severity level. 

We choose $\beta=1$ and observe the linear increase of the aggregated IDoS risk when $m$ is sufficiently large in Fig. \ref{fig:withAM1}. 
We investigate how high and low uncertainty costs $c_{UN}$ affect the aggregated IDoS risk in Fig. \ref{fig:AggregatedriskVSmcUN20} and Fig. \ref{fig:AggregatedriskVSmcUN02}, respectively. 
If the uncertainty cost is much higher than the expected reward of correct decision-making, then the detailed inspection and correct security decisions are not of priority. 
As a result, the `spray and pray' strategy should be adopted; i.e., let the operator inspect as many alerts as possible and use the high quantity to compensate for the low quality of these inspections.  
%Instead, we should let the operator inspect as many alerts as possible to reduce the aggregated IDoS risk; i.e., 
Under this scenario, $\bar{u}(s^{hm},a_m)$ increases with $m\in \mathbb{Z}^+$ for all $s^{hm}\in \mathcal{S}$ as shown in Fig. \ref{fig:AggregatedriskVSmcUN20}. 
If the uncertainty cost is of the same order as the inspection reward on average, then increasing $m$ in a certain range (e.g., $m\in\{1,2,3,4,\}$ in Fig. \ref{fig:AggregatedriskVSmcUN02}) can increase the probability of correct decision-making and reduce the aggregated IDoS risk. 
%inspection correctness is of high priority and 
The loss of alert omissions outweighs the gain of detailed inspection when $m$ is beyond that range. 
\begin{remark}[\textbf{Rational Risk-Reduction Inattention}]
In Fig. \ref{fig:AggregatedriskVSmcUN02}, a small $m$ represents a coarse inspection with a large number of alerts while a large $m$ represents a fine inspection of a small number of alerts. 
The U-shape curve reflects that the minimum risk is achieved with a proper level of \textit{intentional inattention} to alerts, which we refer to as the \textit{law of rational risk-reduction inattention}. 
\end{remark}

\section{Conclusion}
\label{sec:conclusion}
%Innate human vulnerabilities, including bounded rationality and limited attention, has been a great concern
Attentional human vulnerability can be exploited by attackers and leads to a new class of advanced attacks called the Informational Denial-of-Service (IDoS) attacks. 
IDoS attacks intensify the shortage of human operators' cognitive resources in this age of information explosion by generating a large number of feint attacks. 
These feints distract operators from detailed inspections of the alerts, which significantly decrease the accuracy of their security decisions and undermine cybersecurity.
We have formally introduced the IDoS attacks and established a quantitative framework that provides a theoretic underpinning to the IDoS attacks under limited attention resources. 
We have developed human-assistive security technologies that intentionally make selected alerts inconspicuous so that  human operators can pay sustained attention to critical alerts. 

We have modeled the sequential arrival of IDoS attacks as a semi-Markov process and the probability of correct decision-making as an increasing step function concerning the inspection time. 
Dynamic Programming (DP) and Temporal-Difference (TD) learning have been used to represent long-term costs and evaluate human performance in real-time, respectively.  
We have established the \textit{computational equivalency} between the DP representation of the Cumulative Cost (CC) (resp. Expected Cumulative Cost (ECC)) and the Aggregated Cumulative Cost (ACC) (resp. Expected Aggregated Cumulative Cost (EACC)). % when $z$ is independent of $s^k$ and $\theta^k$. 
This equivalency has reduced the dimension of the state space and the computational complexity of the value iteration and online learning algorithms. %has reduced the function dependence from the consolidated state %$x^h\in \mathcal{X}$ 
%to %$s^{hm}\in \mathcal{S}$ 
%and simplified the value iteration and online learning. 
%, and the DP representation of Expected Cumulative Cost (ECC) and Expected Aggregated Cumulative Cost (EACC). 

From the case study, we have validated that both the severity level and the aggregated risk of IDoS attacks increase exponentially with the product of the attack's arrival rate and the operator's inspection efficiency. When Attention Management (AM) strategies are applied, we have observed that the severity level strictly decreases with the inspection time. %and yields a fundamental limit of the minimum severity level that can be achieved under all AM strategies. 
%When attention management strategies are applied, the severity level strictly decreases with the period length $m\in \mathbb{Z}^+$ yet has a limit of $1-\underline{p}(s^{hm})$. 
%The impact of $m$ on the IDoS risk depends on the uncertainty cost. 
We have arrived at the `{less is more}' security principle in cases where correctly identifying the real and feint attacks is of high priority. 
%humans encounter a heavy attentional load. 
It has been shown that inspecting a small number of selected alerts with sustained attention outperforms dividing the limited attention to inspect all alerts. 
%When the cost is small, the results illustrate that the principle of `\textit{less is more}' applies to cybersecurity; i.e., under a certain degree $m\leq \underline{m}(s^{hm})$, it is more beneficial to inspect a small number of alerts with sustained attention than dividing the limited attention to all alerts. 

The future work would focus on coordinating multiple human operators to share the cognition load. 
Based on the literature of cognitive science and existing results of human experiments, we would develop detailed models of human attention, reasoning, and risk-perceiving to better characterize human factors in cybersecurity. 
Finally, we would extend the periodic AM strategies to adaptive ones that use the feedback of the alerts' category labels and the operator's current cognition status reflected by bio-sensors. 

% feedback would be adopted to extend the periodic AM strategies to adaptive ones. 
% For example, we can use bio-sensors, e.g., eye-trackers, to record the current cognition status of human operators. Then, based on the %estimated risk-level of the alerts and 
% human cognition status, the assistive technology would determine whether to highlight the alert upon arrival. 

\bibliographystyle{IEEEtran}
\bibliography{IDoS}

\end{document}